\documentclass{article}

\usepackage{microtype}
\usepackage{graphicx}
\usepackage{booktabs} 

\usepackage{caption}
\usepackage{subcaption}
\usepackage{comment}
\usepackage{tabularx}
\usepackage{amsmath}
\usepackage{amsthm}
\usepackage{amsfonts}
\usepackage{mathtools}
\usepackage{hhline}
\usepackage[dvipsnames]{xcolor}
\usepackage{amsmath}
\usepackage{makecell}
\usepackage{xr-hyper}
\usepackage{hyperref}
\usepackage[shortlabels]{enumitem}

\usepackage{algorithm}
\usepackage[noend]{algpseudocode}
\newcommand*\Let[2]{\State #1 $\gets$ #2}
\newcommand{\algorithmicbreak}{\textbf{break}}

\DeclareMathOperator*{\argmax}{argmax} 

\newtheorem{theorem}{Theorem}

\newtheorem{lemma}{Lemma}

\usepackage[accepted]{icml2021}

\icmltitlerunning{Multi-Agent Training beyond Zero-Sum with Correlated Equilibrium Meta-Solvers}
\begin{document}

\twocolumn[
\icmltitle{Multi-Agent Training beyond Zero-Sum with Correlated Equilibrium Meta-Solvers}


\icmlsetsymbol{equal}{*}

\begin{icmlauthorlist}
\icmlauthor{Luke Marris}{deep,ucl}
\icmlauthor{Paul Muller}{deep,uge}
\icmlauthor{Marc Lanctot}{deep}
\icmlauthor{Karl Tuyls}{deep}
\icmlauthor{Thore Graepel}{deep,ucl}
\end{icmlauthorlist}

\icmlaffiliation{deep}{DeepMind}
\icmlaffiliation{ucl}{University College London}
\icmlaffiliation{uge}{Université Gustave Eiffel}

\icmlcorrespondingauthor{Luke Marris}{marris@google.com}

\icmlkeywords{Machine Learning, ICML}

\vskip 0.3in
]



\printAffiliationsAndNotice{}  

\begin{abstract}
Two-player, constant-sum games are well studied in the literature, but there has been limited progress outside of this setting. We propose Joint Policy-Space Response Oracles (JPSRO), an algorithm for training agents in n-player, general-sum extensive form games, which provably converges to an equilibrium. We further suggest correlated equilibria (CE) as promising meta-solvers, and propose a novel solution concept Maximum Gini Correlated Equilibrium (MGCE), a principled and computationally efficient family of solutions for solving the correlated equilibrium selection problem. We conduct several experiments using CE meta-solvers for JPSRO and demonstrate convergence on n-player, general-sum games.
\end{abstract}

\section{Introduction}

Recent success in tackling two-player, constant-sum games \cite{silver2016_mastering,vinyals2019_starcraft} has outpaced progress in n-player, general-sum games despite a lot of interest \cite{jaderberg2019_ctf,openai2019_dota,brown2019_poker,lockhart2020_bridge,gray2020_nopress,anthony2020_nopress}. One reason is because Nash equilibrium (NE) \cite{nash1951_neq} is tractable and interchangeable in the two-player, constant-sum setting but becomes intractable~\cite{daskalakis2009_ne_complexity} and potentially non-interchangeable\footnote{That is, there are no longer any guarantees on the expected utility when each player plays their part of some equilibrium; guarantees only hold when all players play {\it the same} equilibrium. Since players cannot guarantee what others choose, they cannot optimize independently, so the Nash equilibrium loses its appeal as a prescriptive solution concept.} in n-player and general-sum settings. The problem of selecting from multiple solutions is known as the equilibrium selection problem \cite{goldberg2013_selection_complexity,avis2010_enumeration,harsanyi1988_eq_selection}.\footnote{The equilibrium selection problem is subtle and can have various interpretations. We describe it fully in Section~\ref{sec:properties_selection} based on the classical understanding from ~\cite{harsanyi1988_eq_selection}.}

Outside of normal form (NF) games, this problem setting arises in multi-agent training when dealing with empirical games (also called meta-games), where a game payoff tensor is populated with expected outcomes between agents playing an extensive form (EF) game, for example the StarCraft League \cite{vinyals2019_starcraft} and Policy-Space Response Oracles (PSRO) \cite{lanctot2017_psro}, a recent variant of which reached state-of-the-art results in Stratego Barrage~\cite{mcaleer2020_pipeline}.

In this work we propose using correlated equilibrium (CE) \cite{aumann1974_ce} and coarse correlated equilibrium (CCE) as a suitable target equilibrium space for n-player, general-sum games\footnote{We mean games (also called environments) in a very general sense: extensive form games, multi-agent MDPs and POMDPs (stochastic games), imperfect information games, are all solvable with this approach.}. The (C)CE solution concept has two main benefits over NE; firstly, it provides a mechanism for players to correlate their actions to arrive at mutually higher payoffs and secondly, it is computationally tractable to compute solutions for n-player, general-sum games \cite{daskalakis2009_ne_complexity}. We provide a tractable approach to select from the space of (C)CEs (MG), and a novel training framework that converges to this solution (JPSRO). The result is a set of tools for theoretically solving any complete information\footnote{Payoffs for all players are required for the correlation device.} multi-agent problem. These tools are amenable to scaling approaches; including utilizing reinforcement learning, function approximation, and online solution solvers, however we leave this to future work.

In Section~\ref{sec:preliminaries} we provide background on a) correlated equilibrium (CE), an important generalization of NE, b) coarse correlated equilibrium (CCE)~\cite{moulin1978_cce}, a similar solution concept, and c) PSRO, a powerful multi-agent training algorithm. In Section~\ref{sec:mgce_computation} we propose novel solution concepts called Maximum Gini (Coarse) Correlated Equilibrium (MG(C)CE) and in Section~\ref{sec:properties} we thoroughly explore its properties including tractability, scalability, invariance, and a parameterized family of solutions. In Section~\ref{sec:jpsro} we propose a novel training algorithm, Joint Policy-Space Response Oracles (JPSRO), to train policies on n-player, general-sum extensive form games. JPSRO requires the solution of a meta-game, and we propose using MG(C)CE as a meta-solver. We prove that the resulting algorithm converges to a normal form (C)CE in the extensive form game. In Section~\ref{sec:joint_meta_solvers} we conduct an empirical study and show convergence rates and social welfare across a variety of games including n-player, general-sum, and common-payoff games.

An important area of related work is $\alpha$-Rank \cite{omidshafiei2019_alpharank} which also aims to provide a tractable alternative solution in normal form games. It gives similar solutions to NE in the two-player, constant-sum setting, however it is not directly related to NE or (C)CE. $\alpha$-Rank has also been applied to ranking agents and as a meta-solver for PSRO \cite{muller2020_alpharankpsro}. MG(C)CE is inspired by Maximum Entropy Correlated Equilibria (MECE) \cite{ortix2007_mece}, an entropy maximizing CE based on Shannon's entropy that is harder to compute than Gini impurity.

Another important area of related work concerns optimization based approaches~\cite{stengel2008_efce,miroslav2012_efce,farina2019_coarse}
and no-regret approaches~\cite{celli2019_multiplayer,celli2020_noregret,morrill2020_hindsight}. These approaches identify specific subsets or supersets of (C)CE in the extensive-form game by constructing constraint programs or by local regret minimization using the full representation of the information state space. In contrast, the oracle approach can iteratively identify meta-games with smaller support that summarize the strategic complexity of the game compactly.


\section{Preliminaries}
\label{sec:preliminaries}

This section introduces correlated equilibrium and the multi-agent training algorithm PSRO.

\subsection{Correlated Equilibrium}
\label{sec:preliminaries_ce}

Each player, $p$, in a game has a set of actions $a_p \in \mathcal{A}_p$ (also known as pure strategies) available to it. Let $n$ be the number of players in a game. Let $\mathcal{A} = \otimes_{p} \mathcal{A}_p$ be the joint action space and $a = (a_1, ..., a_n) \in \mathcal{A}$ be a joint action.

Let us index quantities relating to all players apart from player $p$ as $-p = \{1, ..., p-1, p+1, ..., n\}$. Let $\sigma(a) = \sigma(a_p, a_{-p})$ be the probability that joint action $a \in \mathcal{A}$ is played in a game. Let $\sigma(a_p) = \sum_{a_{-p} \in \mathcal{A}_{-p}} \sigma(a_p, a_{-p})$ be the marginal probability of player $p$ taking action $a_p \in \mathcal{A}_p$. Let $\sigma(a_{-p}|a_p)$ be the conditional probability that other players play $a_{-p} \in \mathcal{A}_{-p}$, when player $p$ plays $a_p \in \mathcal{A}_p$. $\sigma$ without arguments should be interpreted as a vector of size $[|\mathcal{A}|]$.

Let $G_p: \mathcal{A} \to \mathbb{R}$ be the payoff function for player $p$ when players play the joint action $a \in \mathcal{A}$. The full game payoff $G$ can therefore be defined by a tensor of shape $[n, |\mathcal{A}_1|, ..., |\mathcal{A}_n|]$. A normal form game is defined by the tuple $\mathcal{G} = (G, \mathcal{A})$.

Define $A_p(a'_p, a_p, a_{-p}) = G_p(a'_p, a_{-p}) - G_p(a_p, a_{-p})$ as the advantage of player $p$ switching action from $a_p$ to $a'_p$, when other players play $a_{-p}$. This can be represented as a matrix, $A_p$, of shape $[|\mathcal{A}_p| (|\mathcal{A}_p| - 1), |\mathcal{A}|]$, since we do not need to compare an action with itself. The matrix is sparse and a fraction of $\frac{1}{\mathcal{A}_p}$ elements are non-zero. We use $A$, with shape $[\sum_p |\mathcal{A}_p| (|\mathcal{A}_p| - 1), |\mathcal{A}|]$, to denote the concatenation of $A_p$ into a two-dimensional matrix.

A correlated equilibrium (CE), is a joint mixed strategy $P(a)$ such that no player $p$ has payoff to gain from unilaterally choosing to play another action $a'_p$ instead of $a_p$. An approximate correlated equilibrium ($\epsilon$-CE)\footnote{There are two competing definitions for approximate CE. We use the computationally convenient one (Section~\ref{supp_subsec:alt_ce}).} is one where that gain from switching actions is no more than $\epsilon$. When $\epsilon=0$, the standard CE is recovered. This relationship is described mathematically in Equation~\eqref{eq:ce_gain_cons}, $\forall p \in \mathcal{P}, a'_p \neq a_p \in \mathcal{A}_p$. In matrix form, we can simply write $A \sigma \leq \epsilon$.
\begin{equation} \label{eq:ce_gain_cons}
    \sum_{a_{-p}} \sigma(a_{-p},a_p) A_p(a'_p, a_p, a_{-p}) \leq \epsilon
\end{equation}

Together, $A_p$ and $\epsilon$ represent the CE linear inequality constraints. Mathematically these are equations of a plane, and separate the joint action space $\sigma(a)$ into half-spaces. Together these half-spaces intersect to form a convex polytope of valid CE solutions.

Of special interest are valid CEs that can factorize into their marginals $\sigma(a) = \prod_p \sigma(a_p)$, and correspond to NE solutions. All NEs are also CEs. Since an NE always exists (when there are finite players and actions) \cite{nash1951_neq}, a CE always exists. An NE is always on the boundary of the polytope for non-trivial games \cite{nau2004_geometry_ce}. Any convex combination of CEs is also a CE.

CEs provide a richer set of solutions than NEs. The maximum sum of social welfare in CEs is at least that of any NE. In particular, this allows more intuitive solutions to anti-coordination games such as chicken and traffic lights. Consider the traffic lights example; a symmetric, general-sum, two-player game consisting of two actions \emph{go}, ($G$), and \emph{wait}, ($W$). $(G, G)$ results in a crash, in $(W, W)$ no progress is made, and $(G, W)$ and $(W, G)$ result in progress for one of the players. Figure~\ref{fig:ce_traffic_lights} shows the NE and CE solution space for the traffic lights game. The mixed NE solution $(G, W) = (\frac{1}{11}, \frac{10}{11})$ is clearly unsatisfactory ($\frac{1}{121}$ crashing and $\frac{100}{121}$ waiting). One could argue that the best solution is to have players flip a coin to decide who waits and who goes. It turns out that this solution is a valid CE and is in fact the unique solution of $\min\epsilon$-MGCE, a novel solution concept that we introduce later in Section~\ref{sec:properties_family}.

\begin{figure*}[t]
\centering

\begin{minipage}[b]{0.40\textwidth}
\centering
\begin{subfigure}[b]{1.0\textwidth}
    \vskip 0pt
    \centering
    \captionsetup{width=.99\linewidth,justification=centering}
    \includegraphics[width=1.0\textwidth]{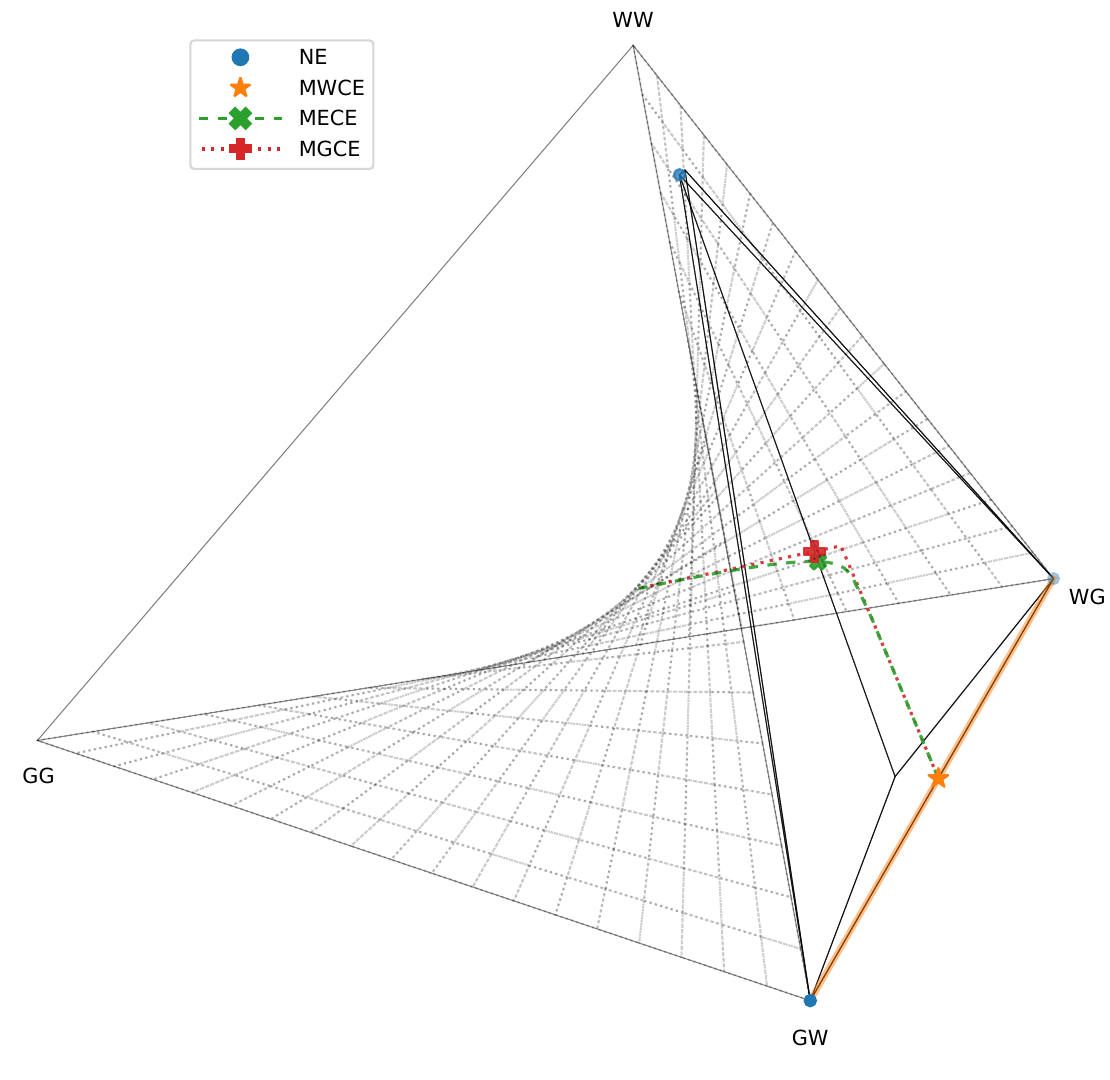}
    \caption{CE Convex Polytope}
    \label{fig:ce_tl_polytope}
\end{subfigure}
\end{minipage}
\begin{minipage}[b]{0.18\textwidth}
\centering
\begin{subfigure}[b]{1.0\textwidth}
    \vskip 0pt
    \centering
    \captionsetup{width=.99\linewidth,justification=centering}
    \begin{tabular}{r|cc}
      ~  & G         & W    \\ \hline
      G  & $-10, -10$ & $1, 0$ \\
      W  & $0, 1$    & $0, 0$
    \end{tabular}
    \caption{Payoff Table \\ \hfill}
    \label{fig:ce_tl_payoff}
\end{subfigure}

\begin{subfigure}[b]{1.0\textwidth}
    \vskip 0pt
    \centering
    \captionsetup{width=.99\linewidth,justification=centering}
    \begin{tabular}{r|cc}
         & G      & W    \\ \hline
      G  & $0.033$ & $0.334$ \\
      W  & $0.334$ & $0.299$
    \end{tabular}
    \caption{MECE \\ $(0,0)$}
    \label{fig:maxent_ce}
\end{subfigure}
\end{minipage}
\begin{minipage}[b]{0.18\textwidth}
\centering
\begin{subfigure}[b]{1.0\textwidth}
    \vskip 0pt
    \centering
    \captionsetup{width=.99\linewidth}
    \begin{tabular}{r|cc}
         & G      & W    \\ \hline
      G  & $\frac{1}{121}$ & $\frac{10}{121}$ \\
      W  & $\frac{10}{121}$ & $\frac{100}{121}$
    \end{tabular}
    \caption{Mixed NE \\ $(0,0)$}
    \label{fig:mixed_ne}
\end{subfigure}

\begin{subfigure}[b]{1.0\textwidth}
    \vskip 0pt
    \centering
    \captionsetup{width=.99\linewidth,justification=centering}
    \begin{tabular}{r|cc}
         & G      & W    \\ \hline
      G  & $0$ & $1$ \\
      W  & $0$ & $0$
    \end{tabular}
    \caption{Row NE \\ $(1,0)$}
    \label{fig:row_ne}
\end{subfigure}

\begin{subfigure}[b]{1.0\textwidth}
    \vskip 0pt
    \centering
    \captionsetup{width=.99\linewidth,justification=centering}
    \begin{tabular}{r|cc}
         & G      & W    \\ \hline
      G  & $0$ & $0$ \\
      W  & $1$ & $0$
    \end{tabular}
    \caption{Col NE \\ $(0,1)$}
    \label{fig:col_ne}
\end{subfigure}
\end{minipage}
\begin{minipage}[b]{0.18\textwidth}
\centering
\begin{subfigure}[b]{1.0\textwidth}
    \vskip 0pt
    \centering
    \captionsetup{width=.99\linewidth,justification=centering}
    \begin{tabular}{r|cc}
         & G      & W    \\ \hline
      G  & $0.033$ & $0.327$ \\
      W  & $0.327$ & $0.313$
    \end{tabular}
    \caption{MGCE \\ $(0,0)$}
    \label{fig:maxgini_ce}
\end{subfigure}

\begin{subfigure}[b]{1.0\textwidth}
    \vskip 0pt
    \centering
    \captionsetup{width=.99\linewidth,justification=centering}
    \begin{tabular}{r|cc}
         & G        & W         \\ \hline
      G  & $0$      & $0.34$    \\
      W  & $0.34$   & $0.32$
    \end{tabular}
    \caption{F-MGCE \\ $(0.34,0.34)$}
    \label{fig:full_maxgini_ce}
\end{subfigure}

\begin{subfigure}[b]{1.0\textwidth}
    \vskip 0pt
    \centering
    \captionsetup{width=.99\linewidth,justification=centering}
    \begin{tabular}{r|cc}
         & G      & W    \\ \hline
      G  & $0$ & $0.5$ \\
      W  & $0.5$ & $0$
    \end{tabular}
    \caption{MWCE \\ $(0.5,0.5)$}
    \label{fig:maxwell_ce}
\end{subfigure}
\end{minipage}

\caption{The solution landscape for the traffic lights game. The solid polytope shows the space of CE joint strategies, and the dotted surface shows factorizable joint strategies. NEs are where the surface and polytope intersect. There are three unsatisfying NEs: mixed spends most of its time waiting and does not avoid crashing, the others favour only the row or column player. One MWCE provides a better solution (note that Row NE and Col NE, and any mixture of the two are also MWCE solutions). The center of the tetrahedron is the uniform distribution and the MECE and MGCE attempt to be near this point. The dashed lines correspond to the family of solutions permitted by MGCE and MECE when varying the approximation parameter $\epsilon$. Both have $(GW, WG) = (0.5, 0.5)$ as the $\min\epsilon$ solution. Player payoffs are given in parenthesis.}
\label{fig:ce_traffic_lights}
\end{figure*}

Correlation is achieved via a trusted external entity (correlation device) which samples a joint action from a public CE joint distribution. Each player is given their action in secret. The properties of the CE means that no individual player is motivated to deviate from the suggested action. If there are deviation actions with equal payoff available, the distribution is a weak equilibrium. If instead, the suggested actions are better than alternatives, the distribution is a strict equilibrium. Distributions that produce actions that are not better than all alternatives are called approximate equilibrium, and the maximum gain that can be obtained by an agent unilaterally deviating from any suggested action is described by $\epsilon > 0$. Weak and strict equilibrium have associated gain $\epsilon=0$ and $\epsilon < 0$, respectively. Mathematically, the effect of reducing $\epsilon$ is shrinking the volume of the CE polytope. We can choose $\epsilon$ when solving for a CE and show in Section~\ref{sec:properties_family} how $\epsilon$ can be used to parameterise a family of MGCE solutions.

There are two important solution concepts in the space of CEs. The first is Maximum Welfare Correlated Equilibrium (MWCE) which is defined as the CE that maximises the sum of all player's payoffs. An MWCE can be obtained by solving a linear program, however the MWCE may not be unique and therefore does not fully solve the equilibrium selection problem (e.g. constant-sum game solutions all have equal payoff). The second such concept is Maximum Entropy Correlated Equilibrium (MECE) \cite{ortix2007_mece} which maximises Shannon's entropy \cite{shannon1948_entropy} as an objective. MECE also shares some interesting properties with MGCE such as computational scalability when the solution is full-support (positive probability mass everywhere). Drawbacks of this approach are that the literature does not provide algorithms when the solution is general-support (non-negative probability) and, maximising Shannon's entropy can be complex.

Finally, coarse correlated equilibrium (CCE)~\cite{moulin1978_cce} (Section~\ref{supp_subsec:cce}) is a simpler solution concept that contains CE as a subset: $\text{NE} \subseteq \text{CE} \subseteq \text{CCE}$. Intuitively, a game distribution is in CCE if no player wishes to deviate \emph{before} receiving a recommended signal.

The solution concepts discussed so far apply to normal form (NF) games, and therefore are sometimes prefixed as such in the literature (NFCE and NFCCE) to disambiguate them from their extensive form (EF) counterparts (EFCE \cite{stengel2008_efce} and EFCCE \cite{farina2019_coarse}). This distinction is important because although EF solutions are a natural choice in EF games; NF solutions can also be applied in EF games by using whole policies $\pi_p \in \Pi_p$ in place of actions $a_p \in \mathcal{A}_p$. These solutions are subsets of one another; $\text{NFCE} \subseteq \text{EFCE} \subseteq \text{EFCCE} \subseteq \text{NFCCE}$ \cite{stengel2008_efce}, therefore NFCE is the most restrictive correlation device while NFCCE is the least restrictive and is therefore capable of achieving the highest welfare. The best correlation device to use is a matter of debate in the literature. However, we note that NF solutions are interesting in EF games because a) it permits the highest welfare, and b) only requires communicating recommendations once before the game starts (as opposed to EF(C)CEs which require communication at every timestep). (J)PSRO trains sets of policies and converges to an NF equilibrium. Therefore, all equilibria discussed in this work are NF and we do not use a prefix going forward.

\subsection{Policy-Space Response Oracles (PSRO)}
\label{subsec:psro}

Policy-Space Response Oracles (PSRO) \cite{lanctot2017_psro} (Algorithm~\ref{alg:psro}) is an iterative population based training method for multi-agent learning that generalizes other well known algorithms such as fictitious play (FP) \cite{brown1951_fp}, fictitious self play (FSP) \cite{heinrich2015_fsp} and double oracle (DO) \cite{mcmahan2003_double_oracle}.

PSRO finds a set of policies, $(\pi_p \in \Pi_p)_{p=1..n}$, and a distribution over this set for each player, $(\sigma_p)_{p=1..n}$. The distribution converges to an NE in two-player, zero-sum games, and has recently been extended to convergence to other types of equilibria \cite{muller2020_alpharankpsro, mcaleer2021_xdo}. This work is in line with these developments, studying convergence of a variant of PSRO with joint policy distributions and (C)CE meta-solvers in n-player, general-sum games.

PSRO consists of a response oracle that estimates the best response (BR) to a joint distribution of policies. Commonly the response oracle is either a reinforcement learning (RL) agent or a method that computes the exact BR. The component that determines the distribution of policies that the oracle responds to is called the meta-solver (MS). The MS operates on the meta-game (MG), which is a payoff tensor estimated by measuring the expected return (ER) of policies against one another. This is a NF game, but instead of strategies corresponding to actions, $a$, they correspond to policies, $\pi$. The set of deterministic policies can be huge and that of stochastic policies is infinite, therefore PSRO only considers a subset of game policies: the ones found by the BR over all iterations so far. Different MSs result in different algorithms: the uniform distribution results in FSP, and using the NE distribution results in an extension of DO.


\section{MG(C)CE and its Computation}
\label{sec:mgce_computation}

The set of (C)CEs forms a convex polytope, and therefore any strictly convex function could uniquely select amongst this set. The literature only provides one such example: MECE \cite{ortix2007_mece} which has a number of appealing properties, but was found to be slow to solve large games. There is a gap in the literature for a more tractable approach, and propose to use the Gini impurity (GI) \cite{breiman1984_cart,bishop2006_pattern}. GI is a member of Tsallis entropy family, a generalized entropy that is equivalent to GI under a certain parameterization. It is maximized when the probability mass function is uniform $\sigma = \frac{1}{|\mathcal{A}|}$ and minimized when all mass is on a single outcome. GI is popular in decision tree classification algorithms because it is easy to compute \cite{breiman1984_cart}. We call the resulting solution concept maximum Gini (coarse) correlated equilibrium (MG(C)CE). This approach has connections to maximum margin \cite{cortes95_svm} and maximum entropy \cite{jaynes1957_maxent}. The derivations (Section~\ref{supp_subsec:scalable_property}) follow standard optimization theory.

\subsection{Quadratic Program}
\label{sec:computation_qp}

The Gini impurity is defined as $1 - \sigma^T\sigma$, and the MG(C)CE is denoted $\sigma^*$. We use an equivalent standard form objective $-\frac{1}{2}\sigma^T\sigma$. The most basic form of the problem can be expressed directly as a quadratic program (QP), consisting of a quadratic objective function (Equation~\eqref{eq:quad_obj}) and linear constraints (Equations \eqref{eq:ce_cons} and \eqref{eq:prob_con}).

\begin{align}
    \text{Gini objective:}& & \max_\sigma - \frac{1}{2} \sigma^T \sigma \quad &\text{s.t.} \label{eq:quad_obj} \\
    \text{(C)CE constraints:}& & A_p \sigma &\leq \epsilon \quad \forall p \label{eq:ce_cons} \\
    \text{Probability constraints:}& & \sigma \geq 0 \quad e^T\sigma &= 1 \label{eq:prob_con}
\end{align}

QPs are a well studied problem class and many techniques may be used to solve them, including convex and quadratic optimization software, such as CVXPY \cite{diamond2016_cvxpy,agrawal2018_cvxpy} and OSQP \cite{stellato2020_osqp}.

\subsection{Primal and Dual Forms}
\label{sec:computation_dual}

The primal objective that we wish to optimize is $\min_\sigma \max_{\alpha, \beta, \lambda} L(\sigma,\alpha, \beta, \lambda) = L_\sigma^{\alpha, \beta, \lambda}$, where $L_\sigma^{\alpha, \beta, \lambda}$ is the primal Lagrangian function, $\alpha_p \geq 0$ are the dual variable vectors corresponding to the $\epsilon-$(C)CE inequality constraints (Equation~\eqref{eq:ce_cons}), $\beta \geq 0$ is the dual variable vector corresponding to the distribution inequality constraints (Equation~\eqref{eq:prob_con}), and $\lambda$ is the dual variable corresponding to the distribution equality constraint (Equation~\eqref{eq:prob_con}). By augmenting the dual variables $\alpha=[\alpha_1, ..., \alpha_n]$ and constraints matrix $A = [A_1, ..., A_n]$, we can write the primal objective compactly as:
\begin{equation} \label{eq:primal_lagrangian_aug}
    L_\sigma^{\alpha_p,\beta,\lambda} = \frac{1}{2} \sigma^T\sigma + \alpha^T (A \sigma - \epsilon) - \beta^T \sigma + \lambda (e^T \sigma - 1),
\end{equation}
where the constant vector of ones with appropriate size is denoted by $e$, and $\epsilon$ is a vector populated with the approximation parameter. We can also formulate a simplified dual version of the optimization as:
\begin{align}
    L^{\alpha,\beta}
    &= - \frac{1}{2} \alpha^T A C A^T \alpha + b^T A^T \alpha - \epsilon^T \alpha -\frac{1}{2} \beta^T C \beta \nonumber \\
    &\quad - b^T \beta + \alpha^T A C \beta + \frac{1}{2} b^T b, \label{eq:dual_lagrangian_aug}
\end{align}
where $C=I - eb^T$ normalizes by the mean, and $b = \frac{1}{|\mathcal{A}|} e$ is the uniform vector. The optimal primal solution $\sigma^*$ can be recovered from the optimal dual variables $\alpha^*_p$ and $\beta^*_p$ using
\begin{equation} \label{eq:primal_from_dual}
    \sigma^* = b - CA^T\alpha^* + C \beta^*.
\end{equation}

The full-support assumption states that all joint probabilities have some positive mass, $\sigma>0$. In this scenario, the dual variable vector corresponding to the non-negative probability constraint is zero, $\beta=0$. Therefore we can define simplified primal and dual objectives.
\begin{align}
    L_\sigma^{\alpha, \lambda} &= \frac{1}{2} \sigma^T\sigma + \alpha^T (A \sigma - \epsilon) + \lambda (e^T \sigma - 1) \label{eq:primal_lagrangian_aug_full} \\
    L^\alpha &= - \frac{1}{2} \alpha^T A C A^T \alpha + b^T A^T \alpha - \epsilon^T \alpha + \frac{1}{2} b^T b \label{eq:dual_lagrangian_aug_full} \\
    \sigma^* &= b - CA^T\alpha^* \label{eq:primal_from_dual_full}
\end{align}


\section{Properties of MG(C)CE}
\label{sec:properties}

In this section we discuss some of the properties of $\epsilon$-MG(C)CE\footnote{Some of the properties discussed here also apply to MECE \cite{ortix2007_mece}.}. Section~\ref{supp_sec:properties_proofs} contains the proofs for this section.

\subsection{Equilibrium Selection Problem}
\label{sec:properties_selection}

There are two levels of coordination; first is selecting an equilibrium before play commences, and second is selecting actions during play time. Both NEs and (C)CEs require agreement on what equilibrium is being played \cite{goldberg2013_selection_complexity,avis2010_enumeration,harsanyi1988_eq_selection}: for (C)CEs this is a joint action probability distribution, and for NEs this is also a joint action probability distribution that can conveniently be factored into stochastic strategies for each player. Therefore, at this level of coordination, both NEs and (C)CEs are similar. We refer to this coordination problem as the \emph{equilibrium selection problem}~\cite{harsanyi1988_eq_selection}. At action selection time only (C)CEs require further coordination. NEs are factorizable and therefore can sample independently without further coordination. (C)CEs rely on a central correlation device that will recommend actions from the equilibrium that was previously agreed upon.

This means that neither NEs nor (C)CEs can be directly used prescriptively in n-player, general-sum games. These solution concepts specify what subsets of joint strategies are in equilibrium, but does not specify how decentralized agents should select amongst these. Furthermore, the presence of a correlation device does not make (C)CEs prescriptive because the agents still need a mechanism to agree on the distribution the correlation device samples from\footnote{This is true if the correlation device is not considered as part of the game. If it was part of the game (for example traffic lights at a junction) the solution concept can appear prescriptive.}. This coordination problem can be cast as one that is more computational in nature: what rules allow an equilibrium to be uniquely (and perhaps de-centrally) selected?

This highlights the main drawback of MW(C)CE which does not select for unique solutions (for example, in constant-sum games all solutions have maximum welfare). One selection criterion for NEs is maximum entropy Nash equilibrium (MENE) \cite{balduzzi2018_nashaverage}, however outside of the two-player constant-sum setting, these are generally not easy to compute \cite{daskalakis2009_ne_complexity}. CEs exist in a convex polytope, so any convex function can select among them. Maximum entropy correlated equilibrium (MECE) \cite{ortix2007_mece} is limited to full-support solutions, which may not exist when $\epsilon=0$, and can be hard to solve in practice. Therefore, there is a gap in the literature for a computationally tractable, unique, solution concept and this work proposes MG(C)CE fills this gap.

\begin{theorem}[Uniqueness and Existence]
MG(C)CE provides a unique solution to the equilibrium solution problem and always exists.
\end{theorem}

\subsection{Scalable Representation}
\label{sec:properties_scalable}

MG(C)CE can provide solutions in general-support and, similar to MECE, MG(C)CE permits a scalable representation when the solution is full-support. Under this scenario, the distribution inequality constraint variables, $\beta$, are inactive, are equal to zero, can be dropped, and the $\alpha$ variables can fully parameterize the solution.

\begin{theorem}[Scalable Representation]
The MG(C)CE, $\sigma^*$, has the following forms:
\begin{align}
    \text{General Support:} \quad \sigma^* &= b - C A^T \alpha^* + C \beta^* \\
    \text{Full Support:} \quad \sigma^* &= b - C A^T \alpha^*
\end{align}
Where $e$ is a vector of ones, $|\mathcal{A}|=\prod_p |\mathcal{A}_p|$, $C = I - e^T b$, and $b=\frac{1}{|\mathcal{A}|}e$ are constants. $\alpha^* \geq 0$ and $\beta^* \geq 0$ are the optimal dual variables of the solution, corresponding to the (C)CE and distribution inequality constraints respectively. 
\end{theorem}

Let $|\mathcal{A}_p|$ correspond to the number of actions available to player $p$, and the total number of joint actions, $\sigma$, is $|\mathcal{A}| = \prod_p |\mathcal{A}_p|$. For each value of $\sigma$, there is a corresponding $\beta$ dual variable. The number of $\alpha$ dual variables is no more than the number of pair permutations $\sum_p |\mathcal{A}_p| (|\mathcal{A}_p| - 1)$ for CEs or actions $\sum_p |\mathcal{A}_p|$ for CCEs. Clearly, games with three or more players and many actions, $\sum_p |\mathcal{A}_p| (|\mathcal{A}_p| - 1) \ll \prod_p |\mathcal{A}_p|$ for CEs and $\sum_p |\mathcal{A}_p| \ll \prod_p |\mathcal{A}_p|$ for CCEs, allow for a very scalable parameterization if the full-support assumption holds. Furthermore, optimal $\alpha^*$ are sparse so we can discard rows from $A$, in a similar spirit to SVMs \cite{cortes95_svm}.

For CEs, full-support is not possible when an action is strictly dominated by another. This case can be easily mitigated by iterated elimination of strictly dominated strategies (IESDS) \cite{fudenberg1991_game_theoy}. This also has the desirable property of simplifying the optimization. In a similar argument, when actions are repeated (having the same payoffs), only one need be retained with appropriate modifications to the optimization.

Among the set of $\epsilon$-MG(C)CE there always exists one with full-support. Note that any infinitesimal positive $\epsilon$ will permit a full-support (C)CE, but $\epsilon$-MG(C)CE does not necessarily select these. An upper bound on $\epsilon$ which permits a full-support solution is given by Theorem~\ref{th:epsilon_full}.
\begin{theorem}[Existence of Full-Support $\epsilon$-MG(C)CE]
For all games, there exists an $\epsilon \leq \max(Ab)$ such that a full-support, $\epsilon$-MG(C)CE exists. A uniform solution, $b$, always exists when $\max(Ab) \leq \epsilon$. When $\epsilon < \max(Ab)$, the solution is non-uniform. \label{th:epsilon_full}
\end{theorem}


\subsection{Family of Solutions}
\label{sec:properties_family}

$\epsilon$-MG(C)CE provides an intuitive way to control the strictness of the equilibrium via the approximation parameter, $\epsilon$, which parameterizes a family of unique solutions. Positive $\epsilon$ expands the solution set and results in a higher Gini impurity solution, at the expense of lower payoff, and approximate equilibrium. Negative $\epsilon$ shrinks the solution set to achieve a strict equilibrium and higher payoff at the expense of Gini impurity. This might also be a more robust solution \cite{wald1939_minimax,wald1945_minimax,ben2009_robust} if the payoff is uncertain.

It is worth emphasizing a set of particularly interesting solutions within this family. Firstly the standard MG(C)CE, with $\epsilon=0$, provides a weak equilibrium for non-trivial games (Theorem~\ref{th:mgce_weak}). Secondly, an edge case with positive $\epsilon$ is $\max(Ab)$-$\epsilon$-MG(C)CE which guarantees a uniform distribution solution. Converging to uniform when increasing $\epsilon$ is a desirable property (principle of insufficient reason) \cite{savage1954_foundations,sinn1980_insufficient_reason,jaynes1957_maxent}. Thirdly, note that all $\epsilon < \max(Ab)$ are guaranteed to have a non-uniform distribution (Theorem~\ref{th:epsilon_full}),  therefore, a $\frac{1}{2}\max(Ab)$-$\epsilon$-MG(C)CE could be an interesting way to regularise a MGCE towards a uniform distribution. Fourthly, because our algorithms are particularly scalable when full-support, working out the minimum $\epsilon$ such that a full-support solution exists, $\text{full}$-$\epsilon$-MG(C)CE, would be useful. Finally, the solution with the smallest feasible $\epsilon$ is the $\min\epsilon$-MG(C)CE. This solution has the lowest entropy of the family, but the highest payoff, and constitutes the strictest equilibrium. Refer to Figure~\ref{fig:ce_traffic_lights} for the family of solutions for the traffic lights game.
\begin{theorem} \label{th:mgce_weak}
    For non-trivial games \cite{nau2004_geometry_ce}, the MG(C)CE lies on the boundary of the polytope and hence is a weak equilibrium.
\end{theorem}

Since the $\epsilon$ is deterministically known for the $\max(Ab)\epsilon$-MG(C)CE, $\frac{1}{2}\max(Ab)\epsilon$-MG(C)CE and MG(C)CE solutions, we can solve for these using the standard solvers discussed in Section~\ref{sec:mgce_computation}. For the $\min\epsilon$-MG(C)CE we can tweak our optimization procedure to solve for this case directly by simply including a $c\epsilon$ term to minimize, where $c>1$. We use bisection search to find $\text{full}$-$\epsilon$-MG(C)CE.

\subsection{Invariance}
\label{sec:properties_invariance}

An important concept in decision theory, called cardinal utility \cite{mascolell1995_micro}, is that offset and positive scale of each player's payoff does not change the properties of the game. A notable solution concept that does not have this property is MW(C)CE.
\begin{theorem}[Affine Payoff Transformation Invariance]
If $\sigma^*$ is the $\epsilon$-MG(C)CE of a game, $\mathcal{G}$, then for each player $p$ independently we can transform the payoff tensors $\tilde{G}_p = c_p G_p + d_p$ and approximation vector $\tilde{\epsilon}_p = a_p \epsilon_p$ for some positive $c_p$ and real $d_p$ scalars, without changing the solution. Furthermore rows of the advantage matrix $A$, and approximation vector, $\epsilon$, can be scaled independently without changing the MG(C)CE.
\end{theorem}

\subsection{Computationally Tractable}
\label{sec:properties_tractable}

In general, finding NEs is a hard problem \cite{daskalakis2009_ne_complexity}. While solving for any valid (C)CE is simple (basic feasible solution of a linear constraint problem) \cite{matousek2006_lp}, and finding a (C)CE with a linear objective is an LP, solving for a particular (C)CE can be hard. For example, MECE \cite{ortix2007_mece} requires optimizing a constrained nonlinear objective. $\alpha$-Rank can be solved in cubic time in the number of pure joint strategies, $O(|\mathcal{A}|^3)$.

MG(C)CE, however, is the solution to a quadratic program, and therefore can be solved in polynomial time. Furthermore, if the assumption is made that the solution is full-support, the algorithm's variables scale better than the number of $\sigma$ parameters.

Space requirements are dominated by the storage of the advantage matrix $A$, which requires a space of $O(n|\mathcal{A}_p| |\mathcal{A}|)$ when exploiting sparsity. Computation is also on the order $O(n|\mathcal{A}_p| |\mathcal{A}|)$ for gradient computation, exploiting sparsity. The number of variables depends on whether we are solving the general-support, $|\mathcal{A}| + n|\mathcal{A}_p|^2$, or full-support, $n|\mathcal{A}_p|^2$ version. It is possible to make use of sparse matrix implementations and only efficient matrix-vector multiplications are required to compute the derivatives.


\section{Joint PSRO}
\label{sec:jpsro}

JPSRO (Algorithm~\ref{alg:jpsro}) is a novel extension to Policy-Space Response Oracles (PSRO) \cite{lanctot2017_psro} (Algorithm~\ref{alg:psro}) with full mixed joint policies to enable coordination among policies. Although a conceptually straightforward extension, careful attention is needed to a) develop suitable best response (BR) operators, b) develop tractable joint distribution meta-solvers (MS), c) evaluate the set of policies found so far, and d) develop convergence proofs.

Using notation of Section~\ref{sec:preliminaries_ce}, but policies instead of actions. Let $(\Pi^*_p)_{p=1..n}$ be the set of all policies of the extensive form game available for each player, and $\Pi^*=\otimes_{p}\Pi^*_p$ be the set of all joint policies. JPSRO is an iteration-based algorithm, let $\{{}^c\pi_p^t, ...\} =\Pi^t_p$ be the set of new policies found at iteration $t$ for player $p$ with $c \in \mathcal{C}$ indexing an individual policy within that set. The set of all policies found so far for player $p$ is denoted $\Pi^{0:t}_p$ and the set of joint policies is denoted $\Pi^{0:t} = \otimes_{p}\Pi^{0:t}_p$. The expected return (ER), an NF game $(G^{0:t}_p)_{p=1..n}$, is tracked for each joint policy found so far such that $G^{0:t}_p(\pi)$ is the expected return to player $p$ when playing joint policy $\pi$. We also define $G^*_p$ to be the payoff over all possible joint policies. 

The MS is a function taking in the ER and returning a joint distribution, $\sigma^t$, over $\Pi^{0:t}$, such that $\sigma^t(\pi)$ is the probability to play joint policy $\pi \in \Pi^{0:t}$ at iteration $t$. The BR operator finds a policy which maximizes the expected return over of opponent mixed joint policies, $\pi_{-p} \in \Pi^{0:t}_{-p}$. This mixture is defined in terms of the MS joint distribution, $\sigma^t$.

\begin{figure*}
\centering
\begin{minipage}[t]{0.48\textwidth}
\begin{algorithm}[H]
    \caption[PSRO]{Two-Player PSRO}
    \label{alg:psro}
    \begin{algorithmic}[1]
        \Let{$\Pi_1^0$, $\Pi_2^0$}{$\{ \pi_1^0 \}$, $\{ \pi_2^0 \}$}
        \Let{$G^0$}{ER($\Pi^0$)}
        \Let{$\sigma_1^0$, $\sigma_2^0$}{MS($G^0$)}
        \For{$t \gets \{ 1, ... \}$}
            \Let{$\pi_1^t$, $\Delta_1^t$}{BR($\Pi_2^{t-1}$, $\sigma_2^{t-1}$)}
            \Let{$\pi_2^t$, $\Delta_2^t$}{BR($\Pi_1^{t-1}$, $\sigma_1^{t-1}$)}
            \Let{$\Pi_1^t$, $\Pi_2^t$}{$\Pi_1^{t-1} \cup \{ \pi_1^t \}$, $\Pi_2^{t-1} \cup \{ \pi_2^t \}$}
            \Let{$G^t$}{ER($\Pi^t$)}
            \Let{$\sigma_1^t$, $\sigma_2^t$}{MS($G^t$)}
            \If{$\Delta_1^t + \Delta_2^t = 0$}
                \State \algorithmicbreak
            \EndIf
        \EndFor
        \Return ($\Pi_1^{0:t}$, $\Pi_2^{0:t}$), ($\sigma_1^t$, $\sigma_2^t$)
    \end{algorithmic}
\end{algorithm}
\end{minipage}
\hfill
\begin{minipage}[t]{0.48\textwidth}
\begin{algorithm}[H]
    \caption[JPSRO]{JPSRO}
    \label{alg:jpsro}
    \begin{algorithmic}[1]
        \Let{$\Pi_1^0$, ..., $\Pi_n^0$}{$\{ \pi_1^0 \}$, ..., $\{ \pi_n^0 \}$}
        \Let{$G^0$}{ER($\Pi^0$)}
        \Let{$\sigma^0$}{MS($G^0$)}
        \For{$t \gets \{ 1, ...\}$}
            \For{$p \gets \{ 1, ..., n \}$}
                \Let{$\{{}^1\pi_p^t, ...\}$, $\{{}^1\Delta_p^t, ...\}$}{BR$_p$($\Pi^{0:t-1}$, $\sigma^{t-1}$)}
                \Let{$\Pi_p^{0:t}$}{$\Pi_p^{0:t-1} \cup \{ {}^1\pi_p^t, ... \}$}
            \EndFor
            \Let{$G^{0:t}$}{ER($\Pi^{0:t}$)}
            \Let{$\sigma^t$}{MS($G^{0:t}$)}
            \If{$\sum_{p,c} {}^c\Delta_p^t = 0$}
                \State \algorithmicbreak
            \EndIf
        \EndFor
        \Return $\Pi^{0:t}$, $\sigma^t$
    \end{algorithmic}
\end{algorithm}
\end{minipage}
\end{figure*}

\subsection{Best Response Operators}
\label{subsec:joint_psro_br}

At iteration $t+1$ each set, $\Pi_p^{0:t}$, can be expanded using either using a CCE or CE best response (BR) operator. The type of BR operator used determines the type of equilibrium that JPSRO converges to (Section \ref{subsec:joint_psro_convergence}).

\paragraph{JPSRO(CCE)} \hfill \\
There is a single BR objective for each player, which expands the player policy set, $\Pi_p^{0:t+1} = \Pi_p^{0:t} \cup \Pi^{t+1}_p$, where $\Pi^{t+1}_p = \{ \text{BR}^{t+1}_p \}$, and $\sigma(\pi_{-p}) = \sum_{\pi_p \in \Pi_p^{0:t}} \sigma(\pi_p, \pi_{-p})$.
\begin{equation} \nonumber
    \text{BR}^{t+1}_p \in \argmax\limits_{\pi^*_p \in \Pi^*_p} \sum_{\pi_{-p} \in \Pi^{0:t}_{-p}} \sigma^t(\pi_{-p}) G^*_p(\pi^*_p, \pi_{-p})
\end{equation}
The CCE BR attempts to exploit the joint distribution with the responder's own policy preferences marginalized out.
    
\paragraph{JPSRO(CE)} \hfill \\
There is a BR for each possible recommendation a player can get, $\Pi_p^{t+1} = \Pi_p^{0:t} \cup \Pi_p^{t+1}$, where $\Pi_p^{t+1} = \{(\text{BR}^{t+1}_p(\pi^i_p))_{i=1..|\Pi_p^{0:t}|}\}$.
\begin{equation} \nonumber
    \text{BR}^{t+1}_p(\pi_p) \in \argmax\limits_{\pi^*_p \in \Pi^*_p} \sum_{\pi_{-p} \in \Pi^{0:t}_{-p}} \sigma^t(\pi_{-p}|\pi_p) G^*_p(\pi^*_p, \pi_{-p})
\end{equation}
Therefore the CE BR attempts to exploit each policy conditional ``slice''. In practice, we only calculate a BR for positive support policies (similar to Rectified Nash \cite{balduzzi2019_rectifiednash}. Computing the $\argmax$ of the BRs can be achieved through RL or exactly traversing the game tree.



\subsection{Meta-Solvers}
\label{subsec:joint_psro_meta_solvers}

We propose that (C)CEs are good candidates as meta-solvers (MSs). They are more tractable than NEs and can enable coordination to maximize payoff between cooperative agents. In particular we propose three flavours of equilibrium MSs. Firstly, greedy (such as MW(C)CE), which select highest payoff equilibria, and attempt to improve further upon them. Secondly, maximum entropy (such as MG(C)CE) attempt to be robust against many policies through spreading weight. Finally, random samplers (such as RV(C)CE) attempt to explore by probing the extreme points of equilibria. Note that these MSs search through the equilibrium subspace, not the full policy space, and this restriction is a powerful way of achieving convergence. Note that since $\text{CEs} \subseteq \text{CCEs}$, one can also use CE MSs with JPSRO(CCE).

\subsection{Evaluation}
\label{subsec:joint_psro_evaluation}

Measuring convergence to NE (NE Gap, \citet{lanctot2017_psro}) is suitable in two-player, constant-sum games. However, it is not rich enough in cooperative settings. We propose to measure convergence to (C)CE ((C)CE Gap in Section~\ref{subsec:metrics}) in the full extensive form game.  A gap, $\Delta$, of zero implies convergence to an equilibrium. We also measure the expected value obtained by each player, because convergence to an equilibrium does not imply a high value. Both gap and value metrics need to be evaluated under a meta-distribution. Using the same distribution as the MS may be unsuitable because MSs do not necessarily result in equilibria, may be random, or may maximize entropy. Therefore we may also want to evaluate under other distributions such as MW(C)CE, because it constitutes an equilibrium and maximizes value. A final relevant measurement is the number of unique polices found over time. The goal of an MS is to expand policy space (by proposing a joint policy to best respond to). If it fails to find novel policies at an acceptable rate, this could be evidence it is not performing well. Not all novel policies are useful, so caution should be exercised when interpreting this metric. If using a (C)CE MS and the gap is positive, it is guaranteed to find a novel BR policy.

\subsection{Convergence to Equilibria}
\label{subsec:joint_psro_convergence}

JPSRO(CCE) converges\footnote{In exponential time in the worst case, however in practice convergence is much faster.} to a CCE and JPSRO(CE) converges to a CE. We provide a sketch of the proofs here, for full proofs see Section~\ref{subsec:jpsro_proof}.

\begin{theorem}[CCE Convergence]\label{theorem:game_cce_convergence}
When using a CCE meta-solver and CCE best response in JPSRO(CCE) the mixed joint policy converges to a CCE under the meta-solver distribution.
\end{theorem}

\begin{theorem}[CE Convergence]\label{theorem:game_ce_convergence}
When using a CE meta-solver and CE best response in JPSRO(CE) the mixed joint policy converges to a CE under the meta-solver distribution. 
\end{theorem}

\begin{proof}
A (C)CE MS provides a distribution that is in equilibrium over the set of joint policies found so far, $\Pi^{0:t}$. For the algorithm to have converged, it needs to also be in equilibrium over the set of all possible joint policies, $\Pi^*$. This is the case when the BR fails to find a novel policy with nonzero gap. Policies that have been found before, by definition of (C)CE, have zero gap. All behavioural policies can be defined in terms of a mixture of deterministic policies. Therefore, given that there are finite deterministic policies the algorithm will converge.
\end{proof}


\section{CEs and CCEs as Joint Meta-Solvers}
\label{sec:joint_meta_solvers}

We evaluate a number of (C)CE MSs in JPSRO on pure competition, pure cooperation, and general-sum games (Section \ref{supp_sec:experiments}). All games used are available in OpenSpiel \cite{lanctot2019_openspiel}. More thorough descriptions of the games used can be found in Section~\ref{supp_sec:environments}. We use an exact BR oracle, and exactly evaluate policies in the meta-game by traversing the game tree to precisely isolate the MS's contribution to the algorithm.

We compare against common MS including uniform, $\alpha$-Rank \cite{omidshafiei2019_alpharank,muller2020_alpharankpsro}, Projected Replicator Dynamics (PRD) \cite{lanctot2017_psro} which is an NE approximator, and random vertex (coarse) correlated equilibrium (RV(C)CE) which randomly selects a solution on the vertices of (C)CE polytope. We also include a random joint and random Dirichlet solvers as baselines. We treat the solutions to the MSs as full joint distributions. Random solvers were evaluated with five seeds and we plot the mean. When evaluating, we measure equilibrium gaps under their own MS distribution and MW(C)CE to provide a consistent and value maximizing comparison. Experiments were ran for up to 6 hours, after which they were terminated.

Kuhn Poker \cite{kuhn1950_poker,southey2009_kuhn_equil,lanctot2014_kuhn_multi} is a zero-sum poker game with only two actions per player. The two-player variant is solvable with PSRO, however the three-player version benefits from JPSRO. The results in Figure~\ref{fig:joint_psro_kuhn} show rapid convergence to equilibrium.

Trade Comm is a two-player, common-payoff trading game, where players attempt to coordinate on a compatible trade. This game is difficult because it requires searching over a large number of policies to find a compatible mapping, and can easily fall into a sub-optimal equilibrium. Figure~\ref{fig:joint_psro_trade_comm} shows a remarkable dominance of CCE MSs. It is clear that traditional PSRO MSs cannot cope with this cooperative setting.

Sheriff \cite{farina2019_sheriff} is a two-player, general-sum negotiation game. It consists of bargaining rounds between a smuggler, who is motivated to import contraband without getting caught, and a sheriff, who is motivated to find contraband or accept bribes. Figure~\ref{fig:joint_psro_sheriff} shows that JPSRO is capable of finding the optimal value.

\begin{figure}[t!]
\centering
\begin{subfigure}[b]{1.0\linewidth}
    \vskip 0pt
    \includegraphics[width=1.\textwidth]{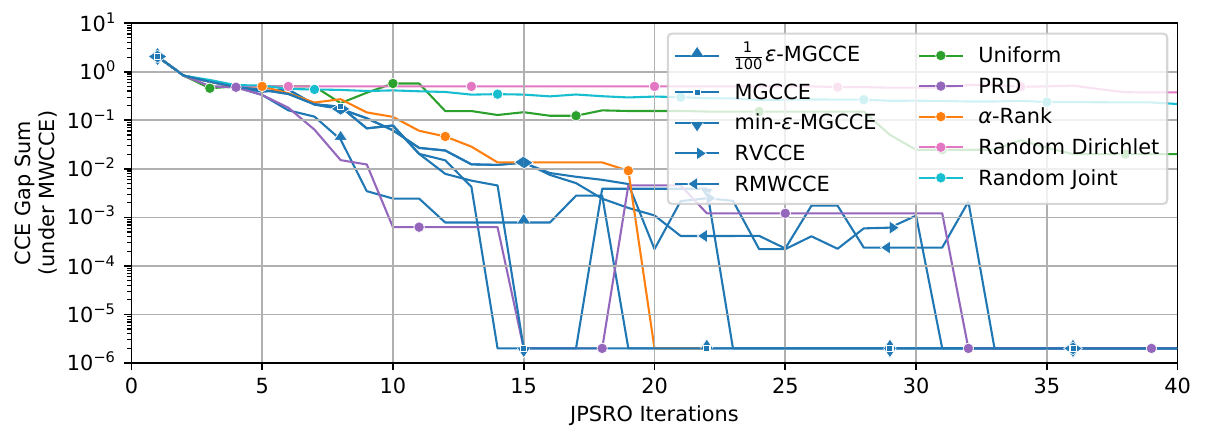}
    \caption{CCE Gap on three-player Kuhn Poker. Several MS converge to within numerical accuracy (data is clipped) of a CCE.}
    \label{fig:joint_psro_kuhn}
\end{subfigure}

\begin{subfigure}[b]{1.0\linewidth}
    \vskip 0pt
    \includegraphics[width=1.\textwidth]{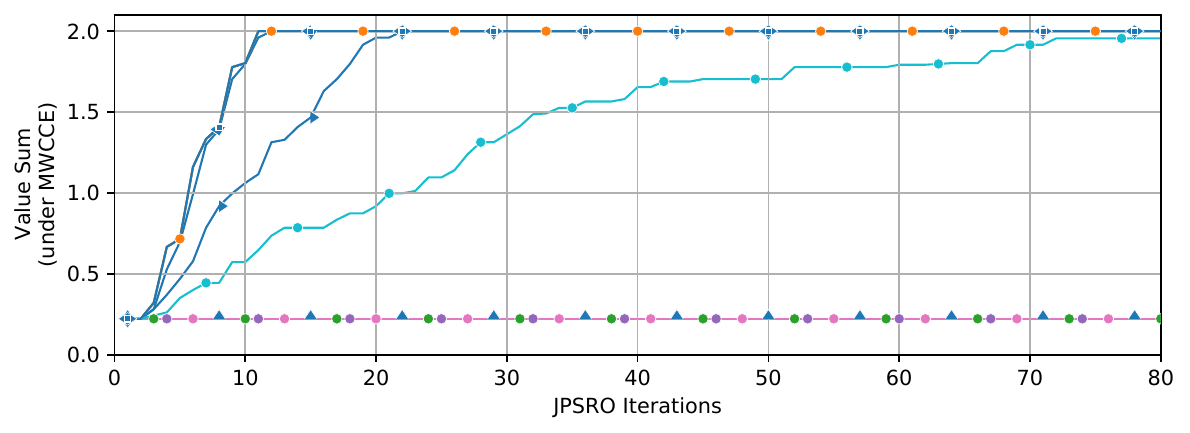}
    \caption{Value sum on three-item Trade Comm. The approximate CCE MS was not sufficient to converge in this game, however all valid CCE MSs were able to converge to the optimal value sum.}
    \label{fig:joint_psro_trade_comm}
\end{subfigure}

\begin{subfigure}[b]{1.0\linewidth}
    \vskip 0pt
    \includegraphics[width=1.\textwidth]{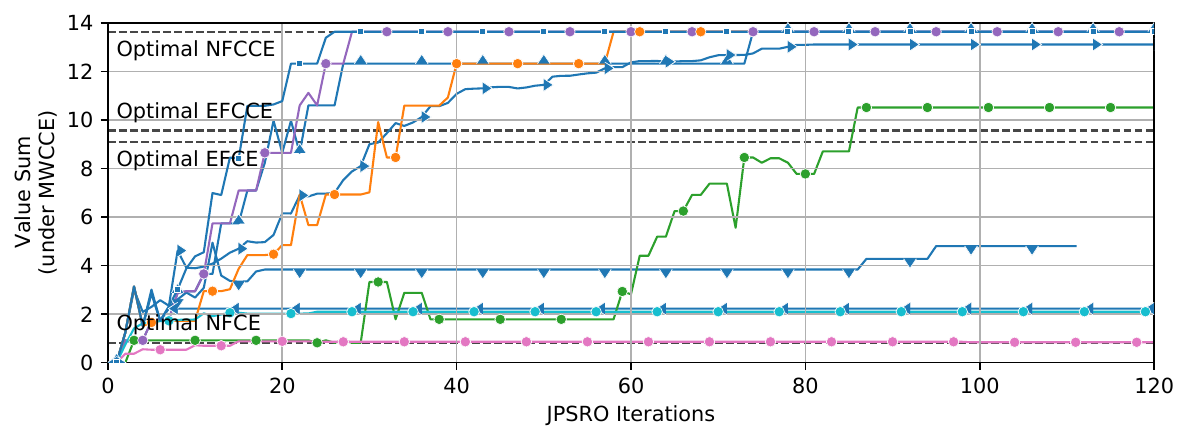}
    \caption{Value sum on Sheriff. The optimal maximum welfare of other solution concepts are included to highlight the appeal of using NFCCE.}
    \label{fig:joint_psro_sheriff}
\end{subfigure}
\caption{JPSRO(CCE) on various games. Additional metrics can be found in Section~\ref{supp_sec:experiments}. MGCCE is consistently a good choice of MS over the games tested.}
\label{fig:joint_psro}
\end{figure}


\section{Discussion}
\label{sec:discussion}



There has been significant recent interest in solving the equilibrium selection problem \cite{ortix2007_mece,omidshafiei2019_alpharank}. This paper provides a novel approach which is computationally tractable, supports general-support solutions, and has favourable scaling properties when the solution is full-support.


The new solution concept MG(C)CE is rooted in the powerful principles of entropy and margin maximisation. Therefore it is a simple solution that makes limited assumptions, and is robust to many possible counter strategies \cite{jaynes1957_maxent}. The MG(C)CE defines a family of unique solutions parameterized by $\epsilon$, that can control for the properties of the distribution. We have compared it to other NE, CE, and $\alpha$-Rank solutions, and have shown it has several advantages over these approaches, and performs very well across a variety of games.

PSRO has proved to be a formidable learning algorithm in two-player, constant-sum games, and JPSRO, with (C)CE MSs, is showing promising results on n-player, general-sum games. The secret to the success of these methods seems to lie in (C)CEs ability to compress the search space of opponent policies to an expressive and non-exploitable subset. For example, no dominated policies are part of CEs, and during execution there are no policies a player would rather deviate to. For (C)CE MSs, if there is a value-improving BR it is guaranteed to be a novel policy.

There is a rich polytope of possible equilibria to choose from, however, an MS must pick one at each time step. There are three competing properties which are important in this regard, exploitation, robustness, and exploration. For exploitation, maximum welfare equilibria appear to be useful. However, to prevent JPSRO from stalling in a local equilibrium it is essential to randomize over multiple solutions satisfying the maximum welfare criterion. To produce robust BRs, entropy maximizing MSs (such as MG(C)CE) have better empirical value and convergence than the uniform MS. For exploration, we can randomly select a valid equilibrium at each iteration which outperforms random joint and random Dirichlet by a significant margin (similar to AlphaStar's ``exploiter policies'' \cite{vinyals2019_starcraft}). Furthermore, one could also switch between MSs at each iteration to achieve the best mix of exploitation and exploration.

Another strength of (C)CE MSs is that they appear to perform well across many different games, with different numbers of players and payoff properties.

\section{Conclusions}

We have shown that JPSRO converges to an NF(C)CE over joint policies in extensive form and stochastic games. Furthermore, there is empirical evidence that some MSs also result in high value equilibria over a variety of games. We argue that (C)CEs are an important concept in evaluating policies in n-player, general-sum games and thoroughly evaluate several MSs. Finally, we believe that both MG(C)CE and JPSRO can scale to large problems, by using stochastic online MSs for the former and exploiting function approximation and RL for the latter.


\section{Acknowledgements}
\label{sec:acknowledgements}

Special thanks to Shayegan Omidshafiei for help with $\alpha$-Rank related discussions, Thomas Anthony for helpful comments and critiques of a draft of the paper, Gabriele Farina for providing maximum welfare values for the Sheriff game, and the anonymous ICML reviewers whose thoughtful feedback strengthened the paper considerably.

\bibliography{bibtex}
\bibliographystyle{icml2021}

\clearpage
\newpage
\appendix
\setcounter{theorem}{0}

\section{Correlated Equilibrium}
\label{supp_sec:ce}

In this section we define the differences between two competing definitions of approximate correlated equilibrium ($\epsilon$-CE) and define the related solution concept approximate coarse correlated equilibrium ($\epsilon$-CCE).

\subsection{Correlated Equilibrium}
\label{supp_subsec:ce}

An approximate correlated equilibrium ($\epsilon$-CE) is one where the  advantage of a single player unilaterally switching away from a recommended action is no more than $\epsilon$. When $\epsilon=0$, the standard CE is recovered. Define $A_p(a'_p, a_p, a_{-p}) = G(a'_p, a_{-p}) - G(a_p, a_{-p})$ as the advantage for player $p$ switching action from $a_p$ to $a'_p$, when other players play $a_{-p}$. This relationship is described mathematically in Equation~\eqref{eq:approx_ce_gain_cons}.
\begin{align}
    \sum_{a_{-p}} \sigma(a_{-p}| a_p) A_p(a'_p, a_p, a_{-p}) &\leq \epsilon_p \label{eq:approx_ce_gain_cons} \\
    \sum_{a_{-p}} \sigma(a_{-p}, a_p) A_p(a'_p, a_p, a_{-p}) &\leq \sigma(a_p) \epsilon_p \label{eq:approx_ce_gain_cons_comp} \\
    \sum_{a_{-p}} \sigma(a_{-p}, a_p) \left ( A_p(a'_p, a_p, a_{-p}) - \epsilon_p \right ) &\leq 0\label{eq:approx_ce_gain_cons_comp2} \\
    \forall p \in \mathcal{P}, a'_p \neq a_p \in \mathcal{A}_p \nonumber
\end{align}

Together, $A_p$ and $\epsilon$ represent the CE linear inequality constraints. Mathematically these are equations of a plane, and separate the mixed joint probability $\sigma(a)$ into half-spaces. Together these half-spaces intersect to form a convex polytope of valid CE solutions.

\subsection{Alternate Form Correlated Equilibrium}
\label{supp_subsec:alt_ce}

Sometimes another definition for CEs is used which is not equivalent to the definition above when $\epsilon \neq 0$. We call Equation~\eqref{eq:approx_ce_gain_cons_other} the alternate approximate CE. In matrix form, we can simply write $A\sigma \leq \epsilon$.
\begin{align} \label{eq:approx_ce_gain_cons_other}
    \sum_{a_{-p}} \sigma(a_{-p}, a_p) A_p(a'_p, a_p, a_{-p}) \leq \epsilon_p \\ \forall p \in \mathcal{P}, a'_p \neq a_p \in \mathcal{A}_p \nonumber
\end{align}

This form is often easier to deal with computationally (particularly with $\min\epsilon$-MG(C)CE) because it does not require dealing with a conditional distribution, and the approximation term is independent of probabilities. We use this definition throughout this work, although MGCE is still well defined using the former definition.

\subsection{Coarse Correlated Equilibrium}
\label{supp_subsec:cce}

The coarse correlated equilibrium (CCE) is a looser solution concept where a player must decide if they are going to play the correlation device's recommendation before they receive the recommendation. Define $A_p(a'_p, a) = G(a'_p, a_{-p}) - G(a)$ as the gain from deviating before an action has been recommended.
\begin{align}
    \sum_{a} \sigma(a) A_p(a'_p, a) &\leq \epsilon_p \\
    \sum_{a} \sigma(a) \left ( A_p(a'_p, a) - \epsilon_p \right ) &\leq 0 \\
    \forall p \in \mathcal{P}, a'_p \in \mathcal{A}_p \nonumber
\end{align}

Note that this can be derived from the correlated equilibrium by mixing over all possible actions, $a_p$, that an agent can take.
\begin{align*}
    \sum_{a_p} \sigma(a_p) \sum_{a_{-p}} \sigma(a_{-p}|a_p) A_p(a'_p, a_p, a_{-p}) &\leq \sum_{a_p} \sigma(a_p) \epsilon_p \\
    \sum_{a_p} \sum_{a_{-p}} \sigma(a_{-p}, a_p) A_p(a'_p, a_p, a_{-p}) &\leq \epsilon_p \\
    \sum_{a} \sigma(a_{-p}, a_p)\left( G(a'_p, a_{-p}) - G(a_p, a_{-p}) \right) &\leq \epsilon_p \\
    \sum_{a} \sigma(a) A_p(a'_p, a) &\leq \epsilon_p
\end{align*}

Note that if one wished to solve for MGCCE, simply substitute the advantage matrix $A^{CE}$ for $A^{CCE}$, without any other additional changes.

\section{Generalized Entropy}
\label{supp_sec:generalized_entropy}

Shannon's Entropy \cite{shannon1948_entropy}, $I_S$, is a familiar quantity and is described as a measure of ``information gain''. The Gini Impurity \cite{breiman1984_cart,bishop2006_pattern} is a measurement of the probability of mis-classifying a sample of a discrete random variable, if that sample were randomly classified according to its own probability mass function, $I_G = \sum_i^N \sigma_i \sum_{j \neq i} \sigma_j = 1 - \sum_i^N \sigma_i^2$.

Both Shannon's entropy and Gini Impurity are maximized when the probability mass function is uniform $\sigma_i = \frac{1}{|\mathcal{A}|}$ and minimized when all mass is on a single outcome. Both metrics are used in decision tree classification algorithms, with Gini being more popular because it is easier to compute \cite{breiman1984_cart}.

In physics, there has been recent interest in non-extensive entropies which have been found to better model certain physical properties. One such entropy is called the Tsallis entropy, $I_T = \frac{1 - \sum_i \sigma_i^q}{q-1}$, \cite{tsallis1988_entropy,havrda1967_alpha_entropy,wang2017_ent_gini_uni,kaur2019_ent_review} and is parameterized by real $q$.

A notable property of the Tsallis entropy is that it is non-additive. Assume that we have two independent variables $A$ and $B$, with joint probability $P(A,B)=P(A)P(B)$, then the combined Tsallis entropy of this system is $I_T(A,B) = I_T(A) + I_T(B) + (1-q) I_T(A) I_T(B)$. Therefore it can be seen that the $(1 - q)$ quantity is a measure of the departure from additivity, with additivity being recovered in the limit when $q \to 1$. This corresponds to the additive Shannon's entropy. The Gini impurity is recovered when $q=2$. Therefore, the Gini impurity is a non-extensive generalized entropy.

\section{Proofs of MG(C)CE Properties}
\label{supp_sec:properties_proofs}

\subsection{Uniqueness and Existence}

\begin{theorem}[Uniqueness and Existence]
MG(C)CE provides a unique solution to the equilibrium solution problem and always exists.
\end{theorem}

\begin{proof}
The problem is a concave maximization problem with linear constraints so therefore has a unique solution. Existence follows from the fact that a CE always exists.
\end{proof}

\subsection{Scalable Representation}
\label{supp_subsec:scalable_property}

\begin{theorem}[Scalable Representation]
The maximum Gini (C)CE, $\sigma^*$, has the following forms:
\begin{align}
    \text{General Support:} \quad \sigma^* &= b - C A^T \alpha^* + C \beta^* \\
    \text{Full Support:} \quad \sigma^* &= b - C A^T \alpha^*
\end{align}
Where $e$ is a vector of ones, $|\mathcal{A}|=\prod_p |\mathcal{A}_p|$, $C = I - e^T b$, and $b=\frac{1}{|\mathcal{A}|}e$ are constants. $\alpha^* \geq 0$ and $\beta^* \geq 0$ are the optimal dual variables of the solution, corresponding to the CE and distribution inequality constraints respectively. 
\end{theorem}

\begin{proof}
Start with the equation we call the primal Lagrangian form.
\begin{align} \label{eq:app_primal_lagrangian}
    L_\sigma^{\alpha, \beta, \lambda} &= \frac{1}{2} \sigma^T\sigma + \alpha (A \sigma - \epsilon) - \beta^T \sigma + \lambda (e^T \sigma - 1)
\end{align}

To construct the dual Lagrangian we first take derivatives with respect to the primal variables $\sigma$, and set them equal to zero.
\begin{align}
    \frac{\partial L_\sigma^{\alpha, \beta, \lambda}}{\partial \sigma} &= \sigma^* + \left( A^T \alpha - \beta + \lambda \right) = 0 \implies \nonumber \\
    \sigma^* &= -A^T \alpha + \beta - \lambda e  \label{eq:app_x_solve}
\end{align}

These can be substituted back into the primal Lagrangian.
\begin{align*}
    L^{\alpha, \beta, \lambda}
    &= -\frac{1}{2}\left[ A^T \alpha - \beta + \lambda e \right]^T \left[ A^T \alpha - \beta + \lambda e \right] \\
    &\qquad - \alpha^T e \epsilon - \lambda
\end{align*}

Taking derivatives with respect to $\lambda$.
\begin{align}
    \frac{\partial L^{\alpha, \beta, \lambda}}{\partial \lambda} &= - |\mathcal{A}| \lambda^* - e^T A^T \alpha_p + e^T \beta - 1 = 0
    \implies  \nonumber  \\
    \lambda^* &= \frac{1}{|\mathcal{A}|} \left ( -e^T A^T \alpha + e^T \beta - 1 \right )  \label{eq:app_lambda_solve}
\end{align}

Substituting $\lambda$ back. Remember that there are non-negative constraints on $\alpha \geq 0$ and $\beta \geq 0$. Therefore, one cannot easily solve for $\beta$ to reduce this expression further. By defining $C = I - e b^T$, and  $b^T = \frac{1}{|\mathcal{A}|} e^T$ (the uniform distribution), noting $b^TC = 0$ and $C^TC = C$, we arrive at the general support dual Lagrangian form.
\begin{align}
    L^{\alpha, \beta}
    &= -\frac{1}{2}\left[ C A^T \alpha - C \beta - \frac{e}{|\mathcal{A}|} \right]^T \left[ C A^T \alpha - C \beta - \frac{e}{|\mathcal{A}|} \right] \nonumber\\
    &\quad - \alpha^T \epsilon + b^T A^T \alpha - b^T \beta + \frac{1}{|\mathcal{A}|} \nonumber \\
    &= - \frac{1}{2} \alpha^T A C A^T \alpha + b^T A^T \alpha - \alpha^T \epsilon  - \frac{1}{2} \beta^T C \beta \nonumber \\
    &\quad - b^T \beta + \alpha^T A C \beta + \frac{1}{2} b^Tb \label{eq:app_dual_lagrangian}
\end{align}

By combining Equations \eqref{eq:app_x_solve} and \eqref{eq:app_lambda_solve}, we can arrive at an equation that describes the relationship between the primal and dual parameters.
\begin{equation} \label{eq:app_primal_from_dual}
    \sigma^* = b - CA^T\alpha^* + C \beta^*
\end{equation}

It is advantageous to try and obtain a more compact representation. We can achieve this if we assume $\sigma$ has full support. In this case, $\beta=0$, because none of the $\sigma\geq0$ constraints are active and we obtain Equation~\eqref{eq:app_dual_lagrangian_full} the full support dual Lagrangian form.
\begin{align} \label{eq:app_dual_lagrangian_full}
    L_\alpha &= - \frac{1}{2} \alpha^T A C A^T \alpha + b^T A^T \alpha - \epsilon^T \alpha + \frac{1}{2}b^Tb \\
    \sigma^* &= b - C A^T \alpha^*
\end{align}
\end{proof}

\begin{theorem}[Existence of Full-Support $\epsilon$-MG(C)CE]
For all games, there exists an $\epsilon \leq \max(Ab)$ such that a full-support, $\epsilon$-MG(C)CE exists. A uniform solution, $b$, always exists when $\max(Ab) \leq \epsilon$. When $\epsilon < \max(Ab)$, the solution is non-uniform.
\end{theorem}

\begin{proof}
Note, $A\sigma \leq \epsilon \iff AC\sigma + Ab \leq \epsilon$, $Cb=0$ and that $b$ is the uniform distribution with maximum possible Gini impurity. Note that when $\max(Ab) \leq \epsilon$ the inequality will always hold with $\sigma=b$. And the inequality cannot hold with $\sigma=b$ when $\epsilon \leq \max(Ab)$.
\end{proof}

\subsection{Family}

\begin{table*}[t!]
\centering
\captionsetup{width=.8\linewidth}
\caption{Family of MG(C)CE solutions.}

\begin{tabular}{rrl}
    MG(C)CE & $\epsilon$ & Properties   \\ \hline
    $\max(Ab)\epsilon$-MG(C)CE            & $\max(Ab)$ & Uniform, highest entropy, lowest payoff \\
    $\frac{1}{2}\max(Ab)\epsilon$-MG(C)CE   & $\frac{1}{2} \max(Ab)$ & Between uniform and (C)CE \\
    $\text{full}\epsilon$-MG(C)CE            & $\leq \max(Ab)$ & Minimum $\epsilon$ such that MG(C)CE is full-support \\
    MG(C)CE            & $0$ & Weak (C)CE, NE in two-player constant sum  \\
    $\min\epsilon$-MG(C)CE             & $\leq 0$ & Strictest (C)CE, lowest entropy, highest payoff
\end{tabular}
\label{tab:ce_family}
\end{table*}

\begin{theorem}
    For non-trivial games, the MG(C)CE lies on the boundary of the polytope and hence is a weak equilibrium.
\end{theorem}

\begin{proof}
    MG(C)CE is attempting to be near the uniform distribution. If the uniform distribution is not a (C)CE the MG(C)CE lies on the boundary of the (C)CE polytope, and by definition is weak. If the uniform distribution is a (C)CE, then it is also a NE (because it factorizes). It therefore lies on the polytope if it is a non-trivial game by \cite{nau2004_geometry_ce}.
\end{proof}

Table \ref{tab:ce_family} summarizes the family of solutions that make up MG(C)CE. Note that a similar family can be defined for ME(C)CE.

\subsection{Invariance}

\begin{theorem}[Affine Payoff Transformation Invariance] \label{theorem:invariance}
If $\sigma^*$ is the $\epsilon$-MG(C)CE of a game, $\mathcal{G}$, then for each player $p$ independently we can transform the payoff tensors $\tilde{G}_p = c_p G_p + d_p$ and approximation vector $\tilde{\epsilon}_p = a_p \epsilon_p$ for some positive $c_p$ and real $d_p$ scalars, without changing the solution. 

Furthermore, if a game, $\mathcal{G}$ has (C)CE constraint matrix, $A$, and bound vector, $\epsilon$, then each row can be scaled independently without changing the MG(C)CE.
\end{theorem}

\begin{proof}
The only way that a game's payoff, $G$, influences the solution is via the (C)CE constraint matrices $A_p$. Recall that these are defined as the difference between action payoffs $a_p \neq a'_p \in \mathcal{A}_p$. It is easy to see that the constant $d_p$ will cancel immediately.
\begin{align}
    \tilde{A}_{p,i,j} &= \tilde{G}_p(a'_p, a_{-p}) - \tilde{G}_p(a_p, a_{-p}) \\
    &= c(G_p(a'_p, a_{-p}) - G_p(a_p, a_{-p})) \nonumber
\end{align}
Notice that $A$ always appears alongside the dual variables $\alpha$. Therefore any scale in $\tilde{A} \tilde{\alpha} = c A \tilde{\alpha}$ can be counteracted by $\tilde{\alpha} = \frac{\alpha}{c}$, without changing the nature of the optimization.

Similar to above, not only does $\alpha_p$ appear alongside $A_p$, each element appears alongside a particular row of $A_p$. Therefore not only can a whole $A_p$ be scaled by a positive factor, each row of $A_p$ can be scaled individually. Intuitively, each row of the (C)CE constraint matrix defines an equation of a plane in the simplex, and planes are not altered when scaled by a positive factor. We may exploit this property to better condition our optimization problem.
\end{proof}

\section{MGCE Computation}
\label{sec:mgce}

There are several tricks that can be employed to simplify the nature of the computation problem.

\subsection{Bounded Gradient Methods}
\label{sec:computation_grad}

It is easy to formulate gradient algorithms to solve for the MG(C)CE. It is most convenient to work in the reduced dual form of the problem as it enforces the probability equality constraint automatically, allows for making the full-support assumption, and does not require any projection routines. The computations involve sparse matrices, so appropriate sparse data structures should be used. The dual variables have a non-negative constraint, which is also sometimes referred to as a box or bound constraints in the literature. 

For gradient ascent, initialize $\alpha^0 = 0$, $\beta^0 = 0$, and update the variables according to their gradient, where $\text{NN}(\sigma) = \max(0, \sigma)$, ensures the variables remain non-negative.
\begin{align}
    \alpha^{t+1} &\gets \text{NN} \left [ \alpha^t + \gamma (-A C A^T \alpha^t + A b - \epsilon + AC \beta^t) \right ] \nonumber \\
    \beta^{t+1} &\gets \text{NN} \left [  \beta^t + \gamma ( - C \beta^t - b + C^T A^T \alpha^t) \right ]
\end{align}

If we assume the solution is full-support, we can simplify the dual version even further by dropping the $\beta$ variable updates.
\begin{align}
    \alpha^{t+1} &\gets \text{NN} \left [ \alpha^t + \gamma (-A C A^T \alpha^t + A b - \epsilon) \right ]
\end{align}

Second order derivatives are also easily computed, allowing use of bounded second order linesearch optimizers, such as L-BFGS-B \cite{byrd1995_lbfgsb}. Other techniques such as momentum \cite{rumelhart1986_backpropagation}, preconditioning the rows of the $A$ matrix, and iterated elimination of strictly dominated strategies of the payoff matrix will also help. An efficient conjugate gradient method can be adapted from Polyak's algorithm \cite{polyak1969_con_cg,oleary1980_cg}, which is a conjugate gradient method modified to support solving problems with bounds and is proven to converge in finite iterations.

\subsection{Payoff Reductions}

There are two methods which could be used to reduce the size of the payoff tensor and hence reduce the complexity of the game that is required to be solved; repeated action elimination, and dominated action elimination.

\begin{description}
    \item[Repeated Action Elimination:] Consider a payoff which has repeated strategies (identical payoffs). This represents a redundancy in the game formulation and we can therefore keep only one of these actions and appropriately modify the objective to account for this alteration. Let $r_p$ be the number of repeats for each action after elimination (i.e. $r_p=e$ if all were unique). Define $r=\otimes_{p}r_p$ as the flattened repeat count which is the same size as $\sigma$ and $\tilde{r}_p = \otimes_{p'}\{ e \text{ if } p' = p \text{ else } r_{p'}\}$. Then the constraints now become $r^T \sigma = 0$ and $A_p(\sigma \cdot \tilde{r}_p) \leq \epsilon_p$, and the objective becomes $1 - \sigma^T (\sigma \cdot r)$. This has the dual effect of reducing the number of variables and constraints in the problem and, more importantly, breaks the symmetry of repeated terms which several solvers can struggle with. It is important to run this procedure before eliminated dominated actions, because repeated actions by definition do not dominate one another.
    
    \item[Dominated Action Elimination:] Strictly dominated strategies can be pruned from the payoff without affecting the results because dominated strategies can never have non-zero support in CEs where $\epsilon \leq 0$. Any CE solution with non-positive $\epsilon$ can exploit this reduction.
\end{description}

The nature of JPSRO means that it is common for actions to be repeated (best responders can produce the same output over multiple distributions) and actions to be strictly dominated by others (as the algorithm finds better and better policies).

\subsection{Eigenvalue Normalization}

Some methods, such as gradient methods, benefit from the eigenvalues of the problem being similar in magnitude. We found empirically that re-normalizing by the $L_2$ norm of the rows of the constraint matrix resulted in eigenvalues close to 1. By Theorem~\ref{theorem:invariance} this is a legal procedure.

\subsection{Dual Optimal Learning Rate}

For the dual form of the objective there is an optimal constant learning rate we can use which is based on the eigenvalues of the Hessian. Calculating the eigenvalues exactly may be too computationally expensive. We can instead obtain an upper bound. A good choice of learning rate that is guaranteed to converge is $\gamma = \frac{2}{\sigma_\text{max} + \sigma_\text{min}^+} \geq \frac{2}{\max_j \sum_{i} |D_{ij}| + \min_j \sum_{i} |D_{ij}|}$, where $D$ is the Hessian of the dual form. A proof follows below.

\begin{proof}
    $C$ is idempotent and positive semi-definite. For any $B$, $BB^T$ is positive semi-definite, therefore $(AC)(AC)^T = ACA^T$ is positive semi-definite. This is the first part of the block diagonals of the Hessian, $D$, which is therefore singular symmetric positive semi-definite.
    
    It is known that the best choice of constant learning rate in this setting is $\gamma=\frac{2}{\sigma_\text{max+} + \sigma_\text{min+}}$. Because the Hessian is not full rank and positive semi-definite,  $\sigma_\text{min} = 0$. We need to find the smallest non-zero eigenvalue. One possible upper bound on the maximal eigenvalues of a positive semi-definitive matrix, by the Gerschgorin circle Theorem \cite{gerschgorin_1931}, is:
    \begin{align}
        \sigma_{max} &\leq \max_j \sum_{i} |D_{ij}| = \max_i \sum_{j} |D_{ij}| \\
        \sigma_{min}^+ &\leq \min_j \sum_{i} |D_{ij}| = \min_i \sum_{j} |D_{ij}|
    \end{align}
\end{proof}

\subsection{$\min\epsilon$-MG(C)CE}
\label{subsec:min_epsilon}

The previous formulations discussed assume that $\epsilon$ is given as a hyper-parameter. If we want to directly find the minimum $\epsilon$ that produces a valid maximum Gini impurity we must also optimize over $\epsilon$. The insight here is that the derivatives of the objective function with respect to the approximation parameter must always be stronger than the derivatives of the objective function with respect to the distribution.
\begin{equation}
    \frac{\partial L}{\partial \epsilon} \geq e^T\frac{\partial L}{\partial \sigma} = -e^T\sigma = -1
\end{equation}

Therefore an additional objective with a term of $-2\epsilon$ would be sufficient to ensure this condition holds.
\begin{align}
    L_\sigma^{\alpha, \beta, \lambda, \epsilon} &= \frac{1}{2} \sigma^T\sigma + 2\epsilon + \alpha^T (A\sigma - \epsilon) \nonumber \\
    &\quad - \beta^T \sigma + \lambda (e^T \sigma - 1) \\
    L^{\alpha, \beta, \epsilon} &= 2 \epsilon - \frac{1}{2} \alpha^T ACA^T \alpha + b^TA^T \alpha - \epsilon^T \alpha \nonumber \\
    &\quad -\frac{1}{2} \beta^T C \beta - b^T \beta + \alpha^T AC \beta + \frac{1}{2} b^Tb \\
    \sigma^* &= b - CA^T\alpha^* + C \beta^*
\end{align}

\section{Joint PSRO}
\label{sec:psro}

While the concept of JPSRO is straightforward, careful attention needs to be made around a) formulating best response operators, b) creating suitable MSs, c) defining evaluation metrics, and d) establishing convergence. We discuss these in detail in this section.

\subsection{Meta Game Estimation}

There are two strategies for estimating the meta-game (a normal form payoff tensor populated by the returns of all the policies); exact sampling and empirical sampling.

\begin{description}
    \item[Exact Sampling:] The exact return is computed for each player by traversing the entire game tree. This is only suitable for small games, or when using deterministic policies that cannot reach the majority of the game tree.
    
    \item[Empirical Sampling:] For larger games, or situations where the policy cannot be easily queried (for example when using a policy that depends on internal state like an LSTM) we may have to  estimate the return through sampling.
\end{description}

In this work we used exact sampling so we could conduct an exact study into the performance of different MSs without introducing noise form other sources. However, the authors believe this approach can be scaled with empirical sampling, as has been achieved with PSRO.

\subsection{Meta-Solvers}

Many of the traditional PSRO solvers are factorizable solutions. Equivalently, their joint probabilities can be marginalized without losing any information.

\begin{description}
    \item[Uniform:] This solver places equal probability mass over each policy it has found so far. PSRO using a uniform distribution is also known as Fictitious Self Play (FSP) \cite{heinrich2015_fsp}. A key advantage of this approach is that it is not necessary to compute the meta-game to obtain this distribution. It is proven to slowly converge in the two-player, constant-sum setting. 
    
    \item[Nash Equilibrium (NE):] The well known solution concept \cite{nash1951_neq}, when used in PSRO is called Double Oracle (DO) \cite{mcmahan2003_double_oracle}. This is difficult to compute for n-player, general-sum, and is equivalent to CCEs in two-player, zero-sum so we did not benchmark against this MS.
    
    \item[Projected Replicator Dynamics (PRD):] An evolutionary method of approximating NE, introduced in \cite{lanctot2017_psro}.
\end{description}

There are a number of solvers which produce full joint distributions. We describe some we think are relevant here. Note that all factorizable solutions mentioned previously can be trivially promoted to full distributions.

\begin{description}
    \item[$\alpha$-Rank:] A solution concept based on the stationary distribution of a Markov chain \cite{omidshafiei2019_alpharank}. $\alpha$-Rank has been studied before in the context of PSRO \cite{muller2020_alpharankpsro}, however the authors marginalized over the distribution.
    
    \item[Maximum Welfare (C)CE (MW(C)CE):] A non-unique linear formulation that maximizes the sum of payoffs over all players. In the case where there are multiple (C)CEs with maximum welfare we can define a maximum entropy version to spread weight, MEMW(C)CE, and a random version to select one at random, RMW(C)CE. We use the latter as an MS baseline in experiments.
    
    \item[Random Vertex (C)CE (RV(C)CE):] A linear formulation. In our implementation we formulate the standard linear (C)CE problem and randomly sample a linear cost function from the unit ball. Note that this selects a random vertex on the (C)CE polytope and is not sampling from within the polytope volume or elsewhere on the polytope surface.
    
    \item[Maximum Entropy (C)CE (ME(C)CE):] A unique nonlinear convex formulation that maximizes the Shannon Entropy of the resulting distribution \cite{ortix2007_mece}. We do not evaluate this solution concept in this work due to computational difficulties when scaling to large payoff tensors, however we expects its performance to be similar to MG(C)CE.
    
    \item[Maximum Gini (C)CE (MGCE):] A unique quadratic convex formulation that maximizes the Gini Impurity (a form of Tsallis Entropy), introduced in this work.
    
    \item[Random Dirichlet:] Sample a distribution randomly from a Dirichlet distribution with $\alpha=1$. This has not been used in the literature before but we believe acts as a good (naive) baseline against RVCE.
    
    \item[Random Joint:] Sample a single joint policy from the set. This has not been used in the literature before either but we believe acts as a good (naive) baseline against RV(C)CE.
    
\end{description}

In previous work joint solvers have been used \cite{muller2020_alpharankpsro}, however the authors marginalized the distributions so they could be used in classic PSRO.

\subsection{Joint Best Responders}
\label{subsec:joint_br}

We provide two best response operators for JPSRO. The first is required to converge to a CCE in policy space (when using CCE meta-solvers). The second is required to converge to a CE in policy space (when using CE meta-solvers).
\begin{description}

    \item[JPSRO(CCE)]: At each iteration there is a single BR objective for each player, which expands the player policy set, $\Pi_p^{0:t+1} = \Pi_p^{0:t} \cup \Pi^{t+1}_p$, where $\Pi^{t+1}_p = \{ \text{BR}^{t+1}_p \}$, and $\sigma(\pi_{-p}) = \sum_{\pi_p \in \Pi_p^{0:t}} \sigma(\pi_p, \pi_{-p})$.
    \begin{equation} \nonumber
        \text{BR}^{t+1}_p \in \argmax\limits_{\pi^*_p \in \Pi^*_p} \sum_{\pi_{-p} \in \Pi^{0:t}_{-p}} \sigma^t(\pi_{-p}) G^*_p(\pi^*_p, \pi_{-p})
    \end{equation}
    Therefore, the CCE BR attempts to exploit the joint distribution with the responder's own policy preferences marginalized out, resulting in a joint policy distribution over the \emph{other} players' policies. This means that a player is best responding to a weighted mixture of up to $\otimes{-p}|\Pi_p^t|$ joint opponent policies. This is an upper bound because $\sigma$ is often sparse.

    \item[JPSRO(CE):] There is a BR for each possible recommendation a player can get, $\Pi_p^{t+1} = \Pi_p^{0:t} \cup \Pi_p^{t+1}$, where $\Pi_p^{t+1} = \{(\text{BR}^{t+1}_p(\pi^i_p))_{i=1..|\Pi_p^{0:t}|}\}$.
    \begin{equation} \nonumber
        \text{BR}^{t+1}_p(\pi_p) \in \argmax\limits_{\pi^*_p \in \Pi^*_p} \sum_{\pi_{-p} \in \Pi^{0:t}_{-p}} \sigma^t(\pi_{-p}|\pi_p) G^*_p(\pi^*_p, \pi_{-p})
    \end{equation}
    Therefore the CE BR attempts to exploit each policy conditional ``slice''. In practice, we only calculate a BR for positive support policies. Computing the $\argmax$ of the BRs can be achieved through RL or exactly traversing the game tree. Similarly each BR is responding to a weighted mixture of up to $\otimes{-p}|\Pi_p^t|$ joint opponent policies.
\end{description}

Notice that if the distribution is factorizable (like NE), then the CE BR is equal for all player policies, and furthermore is equal to the CCE BR, illuminating the connection to PSRO's BR operator.

The best response is independent of the best responding player's policy. We can compute the $\argmax$ in a number of ways. Two common ways are exact best response, and reinforcement learning. 

\begin{description}
    \item[Exact Best Response:] Maintain exact tabular policies and compute a best response against the joint policies for each player, through maximizing value by traversing the game tree. We employ this approach in this work to allow us to compare meta-solvers without introducing noise from approximate BRs. This method is only suitable for small games, or when using only deterministic policies.
    
    \item[RL:] In this setting, the learning algorithms train against randomly sampled joint-policies according to $\sigma$, and do standard value maximization. Both on-policy (such as Policy Gradient) and off-policy (such as Q-Learning) are suitable learning algorithms. Function approximation may also be used. This approach has been used extensively in PSRO before.
\end{description}


\subsection{Evaluation Metrics}
\label{subsec:metrics}

For two-player, constant-sum games there is a clear evaluation metric; how close the players are to the unique Nash Equilibrium (measured by NEGap defined below). However, outside of this narrow setting it is unclear how to fairly evaluate the policies that have been found. This is true for a number of reasons including: there being multiple equilibria, and equilibria not necessarily having good payoff. A combination of high payoff and stability is indicative of a strong set of policies. In this section we describe a number of metrics that could help describe the strength of the resulting joint policies.

\begin{description}
    \item[Value:] This describes the undiscounted return for each player at the root state of a game when following a joint policy, mixed under a joint distribution.
    \begin{align*}
        V_p(\sigma) &= \sum_{\pi \in \Pi} \sigma(\pi) G_p(\pi) = \mathop{\mathbb{E}}\limits_{\pi \sim \sigma} \big [ G_p(\pi) \big ] \\
        V_p(\sigma( \cdotp | \pi_p)) &= \smashoperator{\sum_{\pi_{-p} \in \Pi_{-p}}} \sigma(\pi_{-p} | \pi_p) G_p(\pi_p, \pi_{-p}) \\
        &= \mathop{\mathbb{E}}\limits_{\substack{\pi_{-p} \sim\\\sigma( \cdotp | \pi_p)}} \big [ G_p(\pi_p, \pi_{-p}) \big ] 
    \end{align*}
    
    \item[NE Gap:] This quantity describes how close joint policies are to an NE (referred to as NashConv in \cite{lanctot2017_psro}) under $\sigma$. This is only defined for marginal distributions over policies.
    \begin{align}
        \text{NEGap}_p(\sigma) &= \sum_{\pi \in \Pi} \sigma(\pi) G_p(\text{BR}_p, \pi_{-p}) - V_p(\sigma) \nonumber \\
        &= \mathop{\mathbb{E}}\limits_{\pi \sim \sigma} \big [ G_p(\text{BR}_p, \pi_{-p}) \big ] - V_p(\sigma) \nonumber \\
        \text{NEGap}(\sigma) &= \sum_p \text{NEGap}_p(\sigma)
    \end{align}
    
    \item[CCE Gap:] This quantity describes how close joint policies are to a coarse correlated equilibrium (CCE) under $\sigma$. The origins of this metric can be deduced from studying the CCE BR operator. 
    \begin{align*}
        \text{CCEGap}_p(\sigma) &= \left \lfloor \sum_{\pi \in \Pi} \sigma(\pi) G_p(\text{BR}_p, \pi_{-p}) - V_p(\sigma) \right \rfloor_+ \\
        &= \left \lfloor \mathop{\mathbb{E}}\limits_{\pi \sim \sigma} \big [ G_p(\text{BR}_p, \pi_{-p}) \big ] - V_p(\sigma) \right \rfloor_+ \\
        \text{CCEGap}(\sigma) &= \sum_p \text{CCEGap}_p(\sigma)
    \end{align*}
    Where $\lfloor x \rfloor_+ = max(0, x)$, is the non-negative operator. Note that it is possible for a best response over all joint strategies to have lower value than playing according to the joint distribution for a given player (because a BR is blind to the best responding player's correlation with the opponent policies, and deviating from this correlation can hurt performance).
    
    \item[CE Gap:] This quantity describes how close joint policies are to a correlated equilibrium (CE) under $\sigma$.
    
    \begin{align*}
        \text{CEGap}_p(\sigma, \pi_p) &\quad \\
        &\mkern-120mu= \biggl \lfloor \sum\limits_{\substack{\pi_{-p} \in \\ \Pi_{-p}}} \sigma(\pi_{-p} | \pi_p)  G_p(\text{BR}_p(\pi_p), \pi_{-p}) - V_p(\sigma( \cdotp | \pi_p)) \biggr \rfloor_+  \\
        &\mkern-120mu= \biggl \lfloor \mathop{\mathbb{E}}\limits_{\substack{\pi_{-p} \sim \\ \sigma( \cdotp | \pi_p)}} \big [ G_p(\text{BR}_p(\pi_p), \pi_{-p}) \big ] - V_p(\sigma( \cdotp | \pi_p)) \biggr \rfloor_+  \\
        \text{CEGap}_p(\sigma) &= \sum_{\pi_p \in \Pi_p} \sigma(\pi_p) \text{CEGap}_p(\sigma, \pi_p) \nonumber \\
        \text{CEGap}(\sigma) &= \sum_p \text{CEGap}_p(\sigma)
    \end{align*}
    
    \item[Unique Policy:] Each iteration of JPSRO(CCE) produces n new policies (one for each player), and JPSRO(CE) produces up to the number of policies found so far. These are best responses to the joint mixture of existing polices, however, they are not guaranteed to be distinct from previous policies that have been found. The number of unique policies found so far could be a good indicator of how efficiently a meta-solver is producing new policies.
\end{description}

\subsection{Proof of JPSRO Convergence}
\label{subsec:jpsro_proof}

We provide two convergence proofs for JPSRO. Firstly, when using CCE meta-solvers with a CCE best response operator, which we refer to as JPSRO(CCE), and secondly when using CE meta-solvers with a CE best response operator, which we refer to as JPSRO(CE). Note that, in order to ignore possibly undefined values of $\sigma_t(\pi_{-p} | \pi_p)$, we use the formulation of correlated equilibria using joint probabilities instead of conditional ones. The definitions being equivalent, the conclusions are as well.
Note that we also assume that $\forall p, t, |\text{BR}_p^t| > 0, \forall \pi_p \text{ st. } \sigma_t(\pi_p) > 0, |\text{BR}_p^t(\pi_p)| > 0$, i.e. every time a best response should be computed, it is. We discuss a relaxation of these conditions, and why it is useful, in Section \ref{subsub:proof_relax}.

\subsubsection{Proof of JPSRO(CCE)}
\label{subsubsec:jpsro_cce_proof}

\begin{theorem}[CCE Convergence]
When using a CCE meta-solver and CCE best response in JPSRO(CCE) the mixed joint policy converges to a CCE under the meta-solver distribution.
\end{theorem}

We recall the definition of coarse correlated equilibria. For joint probability $\sigma$, joint policy set $\Pi = \otimes_{p}\Pi_p$ where $\Pi_p$ is the set of valid policies of player $p$ and $\otimes$ is the Cartesian product, and payoff function $G$, such that $G_p(\sigma)$ is the payoff of player $p$ when all player play according to $\sigma$, a Coarse Correlated Equilibrium is a joint distribution $\sigma$ over $\Pi$ such that, for any player $p$ and any policy $\pi'_p$ of player $p$, \begin{equation}\label{eq:cce_eq_joint} \sum\limits_{\pi \in \Pi} \sigma(\pi) G_p(\pi'_p, \pi_{-p}) \leq \sum\limits_{\pi \in \Pi} \sigma(\pi) G_p(\pi) \end{equation}

In other words, a CCE is a distribution from which no player has an incentive to unilaterally deviate \emph{before} being assigned their action. From this definition of CCEs, we derive the definition of \text{CCEGap}, which measures the above gap over all players
\begin{equation*}
    \text{CCEGap}(\sigma) = \sum\limits_p \biggl \lfloor \max_{\pi'_p} \sum\limits_{\pi \in \Pi} \sigma(\pi) ( G_p(\pi'_p, \pi_{-p}) - G_p(\pi)) \biggr \rfloor_+
\end{equation*}
where $\lfloor x \rfloor_+ = max(0, x)$, this $\lfloor \rfloor_+$ term being necessary because the gap is potentially negative, as one can see from Equation~\eqref{eq:cce_eq_joint}. From this definition, we introduce the following lemma:
\begin{lemma}[Game CCE and \text{CCEGap}] \label{lemma:ccegap}
    We have the following equivalence: 
    \begin{enumerate}[(i)]
        \item $\sigma$ is a CCE of the game
        \item CCEGap($\sigma$) = 0
    \end{enumerate}
\end{lemma}

\begin{proof}
Let us first prove (i) $\rightarrow$ (ii). Suppose $\sigma$ is a CCE. Then for any player $p$ and any policy $\pi'_p$ of player $p$,
$$ \sum_{\pi \in \Pi} \sigma(\pi) G_p(\pi'_p, \pi_{-p}) \leq \sum_{\pi \in \Pi} \sigma(\pi) G_p(\pi)$$
therefore, by subtracting the right hand-term and taking the maximum over $\pi'_p \in \Pi_p$,
$$ \max_{\pi'_p} \sum\limits_{\pi \in \Pi} \sigma(\pi) (G_p(\pi'_p, \pi_{-p}) - G_p(\pi)) \leq 0$$
and so
$$
\biggl \lfloor \max_{\pi'_p} \sum\limits_{\pi \in \Pi} \sigma(\pi) (G_p(\pi'_p, \pi_{-p} - G_p(\pi)) \biggr \rfloor_+ = 0
$$
Summing this last inequality over all players yields (ii).

Let us now prove (ii) $\rightarrow$ (i). Suppose that $\sigma$ is such that \text{CCEGap}($\sigma$) = 0. Then, for all $p$, 
\begin{equation}
    \label{eq:cce_diff_lower_zero} 
    \max_{\pi'_p} \sum_{\pi \in \Pi} \sigma(\pi) (G_p(\pi'_p, \pi_{-p}) - G_p(\pi)) \leq 0
\end{equation}For all $\pi''_p \in \Pi_p$ we have $$ \sum\limits_{\pi \in \Pi} \sigma(\pi) G_p(\pi''_p, \pi_{-p}) \leq \max_{\pi'_p} \sum\limits_{\pi \in \Pi} \sigma(\pi) G_p(\pi'_p, \pi_{-p}) $$and therefore, by subtracting $\sum\limits_{\pi \in \Pi} \sigma(\pi) G_p(\pi)$ and using Equation \eqref{eq:cce_diff_lower_zero},$$ \sum\limits_{\pi \in \Pi} \sigma(\pi) (G_p(\pi''_p, \pi_{-p}) - G_p(\pi)) \leq 0 $$Rearranging the terms yields the proof.
\end{proof}

The context of JPSRO motivates us to expand and overload the definition \text{CCEGap}. Let us denote by $\Pi^*$ the policies of the extensive form game, and by $\Pi^{0:t}$ all the policies found by JPSRO by iteration $t$. We immediately have, for all $t$, $\Pi^{0:t} \subset \Pi^*$. We expand \text{CCEGap} via, for all $t$, \begin{align*}
    &\text{CCEGap}(\sigma, \Pi^*, \Pi^{0:t}) = \\
    &\quad \sum_p \biggl \lfloor \max_{\pi^*_p \in \Pi_p^*} \sum\limits_{\pi \in \Pi^{0:t}} \sigma(\pi) (G_p(\pi^*_p, \pi_{-p}) - G_p(\pi)) \biggr \rfloor_+
\end{align*}
The only difference is the search space of $\pi^*_{p}$, which now lives within $\Pi^*$, while the policies used in the sum live in $\Pi^{0:t}$. It is nevertheless easy to see that this new definition characterizes CCEs of $\Pi^*$ (and not of $\Pi^{0:t}$), albeit a restricted class, since $\Pi^{0:t} \subset \Pi^*$ and one can expand $\sigma$ to be zero over $\Pi^* \setminus \Pi^{0:t}$. Let us now prove Theorem~\ref{theorem:game_cce_convergence}.

\begin{proof}
To prove that JPSRO with a CCE meta-solver, JPSRO(CCE), converges to a CCE, we need only prove one thing: that JPSRO(CCE) is unable to produce new policies if and only if it has reached a CCE of the extensive form game. Provided this is true, and since all games have a finite number of deterministic policies, we have that JPSRO(CCE) necessarily cannot produce new policies forever, and therefore eventually can only produce already-discovered policies. 

Note that the joint distribution $\sigma_t$ of JPSRO(CCE) is by construction a CCE over $\Pi^{0:t}$ for all $t$ (when using a CCE meta-solver). It is nevertheless not necessarily a CCE of $\Pi^*$.

Let us now suppose that JPSRO(CCE) has not produced any new policy for any player at iteration $t$. Given the JPSRO(CCE) formulation, we can therefore restrict the search space of policies from $\Pi^*$ to $\Pi^{0:t}$ in the \text{CCEGap} max term, since the max of the expression is reached in $\Pi^{0:t}$, and we thus rewrite the \text{CCEGap} definition:
\begin{align*}
~& \sum_p \biggl \lfloor \max_{\pi'_p \in \Pi^*_p} \sum_{\pi \in \Pi^{0:t}} \sigma_t(\pi) ( G_p(\pi'_p, \pi_{-p}) - G_p(\pi)) \biggr \rfloor_+ \\
=& \sum_p \biggl \lfloor \max_{\pi'_p \in \Pi^{0:t}_p} \sum_{\pi \in \Pi^{0:t}} \sigma_t(\pi) ( G_p(\pi'_p, \pi_{-p}) - G_p(\pi)) \biggr \rfloor_+
\end{align*}
But since $\sigma_t$ is a CCE over $\Pi^{0:t}$, the second term is null. Therefore, $\text{CCEGap}(\sigma, \Pi^*, \Pi^{0:t}) = 0$, and according to Lemma \ref{lemma:ccegap}, $\sigma_t$ is therefore a CCE over $\Pi^*$, which concludes the proof.
\end{proof}


\subsubsection{Proof of JPSRO(CE)}
\label{subsubsec:jpsro_ce_proof}

\begin{theorem}[CE Convergence]
When using a CE meta-solver and CE best response in JPSRO(CE) the mixed joint policy converges to a CE under the meta-solver distribution. 
\end{theorem}

We recall the definition of correlated equilibria. Keeping the same notations as above, a correlated equilibrium is a joint distribution $\sigma$ over $\Pi$ such that, for any player $p$ and any policies $\pi_p$, $\pi'_p$ of player $p$, \begin{align*}
    &\sum\limits_{\pi_{-p} \in \Pi_{-p}} \sigma(\pi_p, \pi_{-p}) G_p(\pi'_p, \pi_{-p}) \leq \\
    &\sum\limits_{\pi_{-p} \in \Pi_{-p}}  \sigma(\pi_p, \pi_{-p}) G_p(\pi_p, \pi_{-p})
\end{align*}

In other words, a CE is a distribution from which no player has an incentive to unilaterally deviate even \emph{after} having been assigned their action. They are therefore stronger than CCEs, and the result CEs $\subseteq$ CCEs easily follows from the above inequality. From this definition of CEs, we derive the definition of \text{CEGap}, which measures the above gap over all players.
\begin{align*}
    &\text{CEGap}(\sigma) = \sum_{p, \pi_p \in \Pi_p} 
    \biggl \lfloor \max_{\pi'_p} \sum\limits_{\pi_{-p} \in \Pi_{-p}} \\ &\quad \sigma(\pi_p, \pi_{-p}) (G_p(\pi'_p, \pi_{-p}) - G_p(\pi_p, \pi_{-p})) \biggr \rfloor_+
\end{align*}

From this definition, we conclude the following lemma:
\begin{lemma}[Game CE and \text{CEGap}]\label{lemma:cegap}
    We have the following equivalence: 
    \begin{enumerate}[(i)]
        \item $\sigma$ is a CE of the game
        \item \text{CEGap}($\sigma$) = 0
    \end{enumerate}
\end{lemma}

\begin{proof}
Let us first prove (i) $\rightarrow$ (ii). Let $\sigma$ be a CE of the game. Therefore, for all $p$, for all $\pi_p, \pi'_p \in \Pi_p$,
\begin{align*} 
    &\sum\limits_{\pi_{-p} \in \Pi_{-p}} \sigma(\pi_p, \pi_{-p}) G_p(\pi'_p, \pi_{-p}) \leq \\
    &\sum\limits_{\pi_{-p} \in \Pi_{-p}} \sigma(\pi_p, \pi_{-p})  G_p(\pi_p, \pi_{-p})
\end{align*}therefore
\begin{align*} 
    \sum\limits_{\substack{\pi_{-p} \in\\\Pi_{-p}}} \sigma(\pi_p, \pi_{-p}) (G_p(\pi'_p, \pi_{-p}) - G_p(\pi_{p}, \pi_{-p})) \leq 0
\end{align*}
which is true for all $\pi'_p \in \Pi_p$, so also true for the max over them
\begin{align*}
    \max_{\pi'_p \in \Pi_p} \sum\limits_{\substack{\pi_{-p} \in\\\Pi_{-p}}} \sigma(\pi_p, \pi_{-p}) (G_p(\pi'_p, \pi_{-p}) - G_p(\pi_p, \pi_{-p})) &\leq 0 \\
    \biggl \lfloor \max_{\pi'_p \in \Pi_{-p}} \sum\limits_{\substack{\pi_{-p} \in\\ \Pi_{-p}}} \sigma(\pi_p, \pi_{-p})  (G_p(\pi'_p,\pi_{-p}) - G_p(\pi_p, \pi_{-p})) \biggr \rfloor_+ &= 0
\end{align*}
Therefore (i) $\rightarrow$ (ii).

Let us now suppose that $\sigma$ is such that $\text{CEGap}(\sigma) = 0$. Thus
\begin{align*}
    &\sum_{p, \pi_p \in \Pi^{0:t}_p{^+}} \biggl \lfloor \max_{\pi'_p} \sum\limits_{\substack{\pi_{-p} \in\\\Pi_{-p}}}
    \sigma(\pi_p, \pi_{-p})\\
    &\qquad (G_p(\pi'_p, \pi_{-p}) - G_p(\pi_p, \pi_{-p}))  \biggr \rfloor_+ = 0
\end{align*}
Given the presence of the positivity operator $\lfloor . \rfloor_+$, we deduce that for all $p$, for all $\pi_p, \pi'_p \in \Pi^{0:t}_p$,
\begin{equation*}
    \sum\limits_{\substack{\pi_{-p} \in\\\Pi_{-p}}} \sigma(\pi_p, \pi_{-p}) (G_p(\pi'_p, \pi_{-p}) - G_p(\pi_p, \pi_{-p})) \leq 0
\end{equation*}
We therefore deduce
\begin{align*}
    &\sum\limits_{\pi_{-p} \in \Pi_{-p}} \sigma(\pi_p, \pi_{-p}) G_p(\pi'_p, \pi_{-p}) \leq \\
    &\sum\limits_{\pi_{-p} \in \Pi_{-p}} \sigma(\pi_p, \pi_{-p}) G_p(\pi_p, \pi_{-p})
\end{align*}
which concludes the proof.
\end{proof}

Once again, the \text{CEGap} definition is extended
\begin{align*}
    \text{CEGap}(\sigma, \Pi^*, \Pi^{0:t}) = \\
    \sum_{p, \pi_p \in \Pi^{0:t}_p} \biggl \lfloor \max_{\pi^*_p \in \Pi_p^*} \sum\limits_{\pi_{-p} \in \Pi_{-p}^t} \sigma(\pi_p,& \pi_{-p}) (G_p(\pi^*_p, \pi_{-p}) - \\ &G_p(\pi_p, \pi_{-p})) \biggr \rfloor_+
\end{align*}
It is once again easy to see that $\text{CEGap}(\sigma, \Pi^*, \Pi^{0:t})$ characterizes CEs of $\Pi^*$.

This lemma proven, we prove Theorem \ref{theorem:game_ce_convergence}.

\begin{proof}
Once again, it is sufficient to prove that JPSRO(CE) stops producing new policies if and only if it has reached a CE of the extensive form game, the rest of the argument being supplied by the finiteness of the game forcing JPSRO(CE) to eventually stop producing new policies.

Let us now suppose that JPSRO(CE) has not produced any new policy for any new player at iteration $t$. This means that for all $\pi_p \in \Pi_p^t$, \begin{align*}
    &\max_{\pi^*_p \in \Pi_p^*} \sum\limits_{\substack{\pi_{-p} \in\\\Pi^{0:t}_{-p}}} \sigma(\pi_p, \pi_{-p}) G_p(\pi^*_p, \pi_{-p}) = \\
    &\max_{\pi'_p \in \Pi_p^t} \sum\limits_{\substack{\pi_{-p} \in\\\Pi^{0:t}_{-p}}} \sigma(\pi_p, \pi_{-p}) G_p(\pi'_p, \pi_{-p})
\end{align*}
We subtract $\sum_{\pi_{-p} \in \Pi_{-p}^t} \sigma(\pi_p, \pi_{-p}) G_p(\pi_p, \pi_{-p})$ to both expressions, apply $\lfloor . \rfloor_+$ and sum over $\pi_p \in \Pi_p^t$ and $p$, and finally apply the fact that $\sigma$ is a CE of the restricted game to obtain that
\begin{align*}
    &\text{CEGap}(\sigma, \Pi^*, \Pi^{0:t}) = \sum_{p, \pi_p \in \Pi_p} \biggl \lfloor \max_{\pi'_p \in \Pi_p^t} \sum\limits_{\pi_{-p} \in \Pi_{-p}} \\ &\sigma(\pi_p, \pi_{-p}) (G_p(\pi'_p, \pi_{-p}) - G_p(\pi_{p}, \pi_{-p})) \biggr \rfloor_+ = 0
\end{align*}
which, by extension, is also true for the \text{CEGap} over the extensive form game. By Lemma \ref{lemma:cegap}, $\sigma$ is therefore a CE of the extensive form game, which concludes the proof.
\end{proof}

\subsubsection{Relaxation on Proof Requirements \label{subsub:proof_relax}} 

Our definition of Best Responses (BRs) is that they are functions that return a set of policies which maximize their value against a given objective. There are two reasons to add a set of policies. Firstly, the max of a given objective can be reached at different points, thus returning a set of policies enables us to potentially include them all. Secondly, using sets also enables us to potentially set some of the BR outputs to $\emptyset$. Concretely, this means that no policy is computed by the BR in that case, which saves compute time and memory. The proofs shown so far rely on each BR having cardinality greater than or equal to 1, which means that one should compute at least one new policy every time the BR operator is called. We can relax this condition into the following conditions, which we prove are sufficient (but not necessary) for convergence.

\begin{description}
    \item[CCE-Condition:]
    \begin{equation*}
        \forall T > 0, p, \exists t > T, |\text{BR}^t_p| \geq 1\end{equation*}
    i.e. each player receives an infinity of best responses. 
    
    \item[CE-Condition:]
    \begin{align*}
        \forall T > 0, p, \pi_p, \exists t > T, \text{ either } & \forall t' \geq t, \sigma_{t'}(\pi_p) = 0 \ \\ \text{ or } & |\text{BR}^t_p(\pi_p)| \geq 1
    \end{align*}
    i.e. any policy of any player is either never selected by the CE meta-solver after some time, or is considered for a best response an infinite number of times.
    
    \item[Solver-Condition:] $\forall t, \forall t' \geq t$, if $\forall p, \forall \pi_p \in \Pi^{0:t'}_p, \pi_p \in \Pi^{0:t}_p$, then $\forall \pi \in \Pi^{0:t}$ (or $\pi \in \Pi^{t'}$), $\sigma_t(\pi) = \sigma_{t'}(\pi)$: if no new policy has been added to the pool between $t$ and $t'$, the amount of mass granted to each policy by the solver does not change, i.e. repeating policies does not affect solver outputs, and the solver's outputs are constant given the same pools.
\end{description}

The rest of this section presents the relaxed theorems, their proofs, and discusses why such a relaxation is of interest.

\paragraph{Relaxed Theorems and Proofs} \hfill \\
\begin{theorem}[Relaxed CCE-Convergence]
When using a CCE meta-solver and CCE best response in JPSRO(CCE), under CCE-Condition and Solver-Condition, the mixed joint policy converges to a CCE under the meta-solver distribution.
\end{theorem}
\begin{proof}
Let us suppose CCE-Condition and Solver-Condition. We have that JPSRO(CCE) will necessarily be able to produce new policies until it reaches a CCE. Let us prove this: while $\text{CCEGap}(\sigma_t, \Pi^*, \Pi^{0:t}) > 0$, JPSRO(CCE) is able to add at least one new policy to its pool. Indeed, let $t > 0$ be such that $\text{CCEGap}(\sigma_t, \Pi^*, \Pi^{0:t}) > 0$. Then there exists at least one $p$ such that
\begin{equation*}
    \max_{\pi'_p \in \Pi^*_p} \sum_{\pi \in \Pi^{0:t}} \sigma_t(\pi) ( G_p(\pi'_p, \pi_{-p}) - G_p(\pi)) > 0.
\end{equation*}
Let us select one of these $p$ with minimal $t' \geq t, |\text{BR}_p^t| \geq 1$, i.e. the first best response with positive \text{CCEGap} to be added to the pool after and including $t$. $t'$ exists because we suppose CCE-Condition. Let us suppose that no new policies have been added to the pool between $t$ and $t'$. Then, since no new best response has been added to the pool between $t$ and $t'$, $\sigma_t = \sigma_{t'}$ since we suppose Solver-Condition, and therefore $\forall \pi' \in \text{BR}_p^{t'}$,
\begin{equation*}
    \sum_{\pi \in \Pi^{0:t}} \sigma_t(\pi) ( G_p(\pi'_p, \pi_{-p}) - G_p(\pi)) > 0.
\end{equation*}We have that necessarily, $\text{BR}_p^{t'} \cap \Pi^{0:t}_p = \emptyset$, as otherwise $\sigma_t$ would not be a CCE of $\Pi^{0:t}$: indeed, since $\sigma_t$ is a CCE of $\Pi^{0:t}$, $\text{CCEGap}(\sigma_t, \Pi^*, \Pi^{0:t}) = 0$, and thus $\forall p, \pi'_p \in \Pi^{0:t}_p$,
\begin{equation*}
    \sum_{\pi \in \Pi^{0:t}} \sigma_t(\pi) ( G_p(\pi'_p, \pi_{-p}) - G_p(\pi)) \leq 0,
\end{equation*}
thus new best responses can be added to the pool. We therefore have that $\text{CCEGap}(\sigma_t, \Pi^*, \Pi^{0:t}) > 0$ implies that at least one new policy can be found by JPSRO.

Thus a new best response can always be added, and will always be added since we have CCE-Condition, to the pool while $\sigma_t$ is not a CCE of the extensive form game. Therefore, if JPSRO(CCE) is unable to add any new policy to the pool (which has to be verified over all players, or measured through $\text{CCEGap}$), then it must be at a CCE, which concludes the proof.
\end{proof}

\begin{theorem}[Relaxed CE-Convergence]
When using a CE meta-solver and CE best response in JPSRO(CE), under CE-Condition and Solver-Condition, the mixed joint policy converges to a CE under the meta-solver distribution. 
\end{theorem}
\begin{proof}
Let us suppose CE-Condition and Solver-Condition. We have that JPSRO(CE) will necessarily be able to produce new policies until it reaches a CE. Let us prove this: while $\text{CEGap}(\sigma_t, \Pi^*, \Pi^{0:t}) > 0$, JPSRO(CE) is able to add at least one new policy to its pool. Indeed, let $t > 0$ be such that $\text{CEGap}(\sigma_t, \Pi^*, \Pi^{0:t}) > 0$. Then there exists at least one $p, \pi_p \text{ st. } \sigma_t(\pi_p) > 0$ such that
\begin{equation*}
    \max_{\pi'_p \in \Pi^*_p} \sum_{\pi_{-p} \in \Pi_{-p}^t} \sigma_t(\pi_p, \pi_{-p}) ( G_p(\pi'_p, \pi_{-p}) - G_p(\pi_p, \pi_{-p})) > 0.
\end{equation*}

By CE-Condition, we have that either new policies have been added to the pool before any such $p, \pi_p$ has been selected, or that there exists $t'$ such that $t' \geq t, |\text{BR}_p^t(\pi_p)| \geq 1$. Indeed, if no new best response has been added to the pool by $t' \geq t$, the Solver-Condition implies that for all these $p, \pi_p \text{ st. } \sigma_t(\pi_p) > 0$, we also have $\sigma_{t'}(\pi_p) > 0$, hence there exists $t'$, $|\text{BR}^t_p(\pi_p)| > 1$. Let us select the minimal $t'$ over all $p, \pi_p$ such that $\text{CEGap}_p(\sigma_t, \Pi^*, \Pi^{0:t})(\pi_p) > 0$.

Let us suppose that no new policies have been added to the pool between $t$ and $t'$. Then, since no new best response has been added to the pool between $t$ and $t'$, $\sigma_t = \sigma_{t'}$ since we suppose Solver-Condition, and therefore $\forall \pi' \in \text{BR}_p^{t'}(\pi_p), \sum_{\pi_{-p} \in \Pi_{-p}^t} \sigma_t(\pi_p, \pi_{-p}) ( G_p(\pi'_p, \pi_{-p}) - G_p(\pi_p, \pi_{-p})) > 0$. We have that necessarily, $\text{BR}_p^{t'}(\pi_p) \cap \Pi^{0:t}_p = \emptyset$, as otherwise $\sigma_t$ would not be a CE of $\Pi^{0:t}$: indeed, since $\sigma_t$ is a CE of $\Pi^{0:t}$, $\text{CEGap}(\sigma_t, \Pi^*, \Pi^{0:t}) = 0$, and thus $\forall p, \pi_p \in \Pi^{0:t}_p, \pi'_p \in \Pi^{0:t}_p$,
\begin{equation*}
    \sum_{\pi_{-p} \in \Pi^{0:t}_{-p}} \sigma_t(\pi_p, \pi_{-p}) ( G_p(\pi'_p, \pi_{-p}) - G_p(\pi_p, \pi_{-p})) \leq 0.
\end{equation*}
Thus new best responses can be added to the pool. We therefore have that $\text{CEGap}(\sigma_t, \Pi^*, \Pi^{0:t}) > 0$ implies that at least one new policy can be found by JPSRO.

Thus a new best response can always be added, and will always be added since we have CE-Condition, to the pool while $\sigma_t$ is not a CE of the extensive form game. Therefore, if JPSRO(CE) is unable to add any new policy to the pool (Which has to be verified over all players, or measured through $\text{CEGap}$), then it must be at a CE, which concludes the proof.
\end{proof}

\paragraph{Discussion on Relaxation} \hfill \\
These relaxed conditions matter especially for JPSRO(CE), which has potentially exponential complexity in term of number of policies to keep (if the solver spreads mass on all policies at each iteration, then the number of policies in each players' pools at iteration $t$ is $\geq 1 + \sum_{k=1}^t 2^k = 2^{t+1}-1$).

Given that the policies produced for one player at the same iteration are potentially similar (even identical), a number of modifications could be imagined to keep JPSRO(CE) tractable. For example: a) randomly select only one $\pi_p$ from which to best respond for each player, b) only compute a best response for one randomly chosen $\pi_p$, or c) compute all BRs, but only add the BR with the largest gap to the pool.

It could make sense to randomly select only one $\pi_p$ from which to best respond for each player, at each iteration, or even to only compute a best response for one randomly chosen $\pi_p$ for one randomly-chosen $p$ at each iteration.

Note that it is necessary to impose a condition on the solver (although an alternate Solver-Condition could be formulated). To illustrate this, let us imagine modes between the best response chooser and the solver. Namely, let us imagine a two-player game, for which on even $t$, in JPSRO(CCE), the best response operator only computes one best response for player 1 (and on odd $t$, the best response is computed only for player 2). Let us also infer that the current restricted game has two CCEs. The first of these (CCE1) is not ``expandable'' for player 1, but is for player 2 (i.e. the best response for player 1 is already in the pool, but player 2's best response is not). The second (CCE2) is expandable for player 1, but not for player 2. If the CCE solver outputs CCE1 on even $t$, and CCE2 on odd $t$, then the algorithm never produces new policies, and therefore never converges.

Of course, the conditions provided are sufficient, but not necessary, and in the case where best response and meta-solver outputs' randomizations are decorrelated, it makes intuitive sense that the algorithm should also converge with probability 1, which one can prove with a more involved argument.

\section{Games}
\label{supp_sec:environments}

We study several games with JPSRO; Kuhn Poker, Trade Comm, and Sheriff. These cover three-player, general-sum, and common-payoff games. Implementations of all the games are available in OpenSpiel \cite{lanctot2019_openspiel}.

\begin{description}
    \item[Kuhn Poker:] A simplified n-player, zero-sum, sequential, imperfect information version of poker. It consists of $n+1$ playing cards. In each round of the game, every player remaining \emph{antes} one chip. One card is dealt to each player. Each player has two choices, \emph{bet} one chip or \emph{check}. If a player bets other players have the option to \emph{call} or \emph{fold}. Out of the players that bet, the one with the highest card wins. If all players check the player with the highest card wins. The original two-player game is described in \cite{kuhn1950_poker}. An n-player extension is described in \cite{lanctot2014_kuhn_multi}. Additional information about the game (such as equilibrium) can be found in \cite{hoehn2005_kuhn_info}.
    
    \item[Trade Comm:] A simple two-player, common-payoff trading game \cite{sokota2021_commonpayoff}. In this game each player (in secret) receives one of $I$ different items. The first player can then make one of $I$ utterances to the second agent, and vice versa. Then each agent chooses one of $I^2$ trades in private, if the trade is compatible both agents receive $1$ reward, otherwise both receive $0$. The goal of the agents is therefore to find a bijection between the items and utterances and the trade proposal. There are $I^4$ deterministic policies per player, and good learning algorithms will be be able to search over these policies. Because the game is common-payoff, it is very transitive, and has many dominated strategies, however there are multiple strategies with equal payoff, and therefore many equilibria in partially explored policy space. It is for this reason many learning algorithms get stuck exploiting sub-optimal policies they have already found.
    
    \item[Sheriff:] A simplified two-player, general-sum version of the board game Sheriff of Nottingham \cite{farina2019_sheriff}. The game consists of a smuggler, who is motivated to import contraband without getting caught, and a sheriff, who is motivated to either find contraband or accept bribes. The players negotiate a bribe over several rounds after which the bribe if accepted or rejected. If the sheriff finds contraband, the smuggler pays a fine, otherwise if no contraband is found the sheriff must pay compensation to the smuggler. The smuggler also gets value from smuggling goods. The game has different optimal values for NFCCE, EFCCE, EFCE, and NFCE solutions concepts.
\end{description}

\section{JPSRO Hyper-parameters}

There are a number of ways of implementing JPSRO in practice through various hyper-parameters.

\begin{description}
    \item[Best Response:] We use an exact best response calculation that assigns uniform probability over valid actions for states with zero reach probability. However, other best response approaches will also work including reinforcement learning (which we will leave to future work).

    \item[Pool Type:] The data structure used to store the policies found so far can either be a set or a multi-set. Using a set ensures that all policies are unique and only appear once even if multiple iterations produce the same best response policy. Some meta-solvers rely on repeated policies being present for convergence (for example, the uniform meta-solver can converge in two-player, zero-sum because the repeated policies trend to a NE over repeats). In this case using a multi-set is more suitable. This parameterization is only relevant when using tabular policies which can be checked for equality.
    
    \item[Player Updates Per Iteration:] It is not necessary to find the best response for all players at every iteration. Other strategies such as cycling through players or randomly selecting a player will work too. It is sufficient that over time all players should be updated. Updating a single player at a time is more efficient when minimizing the number of best responses necessary for convergence, however updating all can be done in parallel.
    
    \item[Best Responses Per Iteration:] When computing the CE best response, each player has several best responses to calculate. It is not necessary to compute them all and, even if they are all computed, it is not necessary to add them all to the pool of policies. The best responses can be calculated at random. And only best responses with nonzero gap need be added, or perhaps only the one with largest gap. In order to measure convergence to a CE, all best responses (and their gaps) must be computed. 
    
    \item[Policy Initialization:] Policies can be initialized in any manner and the algorithm will converge to an equilibrium under any initial condition. However, the initial policies does determine the space of equilibrium reachable (so for example is may not be possible to find the MWCE from all initial policies). JPSRO works, without limitation, using only deterministic policies, however stochastic policies are supported too. A stochastic uniform policy over valid actions is a reasonable setting.
    
    \item[Best Response Type:] The most important parameterization is picking one of the two best response types: CE and CCE. The resulting algorithm is named either JPSRO(CE) or JPSRO(CCE) respectively.
    
    \item[Meta-Solvers:] The second most important parameterization is the type of meta-solver to use (Table~\ref{tab:meta_solvers}). An important constraint is that JPSRO(CE) is only guaranteed to converge under CE meta-solvers. JPSRO(CCE) must use CCE meta-solvers (noting that CEs are a subset of CCEs).

\end{description}

\begin{table}[t]
\centering
\caption{Summary of meta-solvers used during experiments and their properties. We use the normalized $\epsilon$ for naming, for example $\frac{1}{100}\epsilon$-MGCE means $\frac{1}{100}\max(Ab)\epsilon$-MGCE.}

\scriptsize
\addtolength{\tabcolsep}{-1pt} 
\begin{tabular}{rcccccc}
\hline
\thead{Meta-Solver} & \thead{Joint} & \thead{CCE} & \thead{CE} & \thead{Max \\ Val} & \thead{Max \\ Ent} & \thead{Rand} \\ \hline
Uniform &  &  &  & & \checkmark &  \\
PRD &  &  &  &  &  &  \\
$\alpha$-Rank & \checkmark &  &  &  &  &  \\ \hline
Rand Dirichlet & \checkmark &  &  & & &  \checkmark \\
Rand Joint & \checkmark &  &  &  & &  \checkmark \\ \hline
RMWCCE & \checkmark & \checkmark &  & \checkmark & &  \checkmark \\
RVCCE & \checkmark & \checkmark &  &  & &  \checkmark \\
$\frac{1}{100}\epsilon$-MGCCE & \checkmark & $\epsilon$ &  & & \checkmark &   \\
MGCCE & \checkmark & \checkmark &  & & \checkmark &  \\
$\min\epsilon$-MGCCE & \checkmark & \checkmark &  & & \checkmark &   \\ \hline
RMWCE  & \checkmark & \checkmark & \checkmark  & \checkmark & &  \checkmark \\
RVCE & \checkmark & \checkmark & \checkmark  &  & &  \checkmark \\
$\frac{1}{100}\epsilon$-MGCE & \checkmark & $\epsilon$ & $\epsilon$  &  & \checkmark &   \\
MGCE & \checkmark & \checkmark & \checkmark  & & \checkmark &   \\
$\min\epsilon$-MGCE & \checkmark & \checkmark & \checkmark  &  & \checkmark &
\end{tabular}
\addtolength{\tabcolsep}{1pt}
\label{tab:meta_solvers}
\vskip -10pt
\end{table}

\section{Experiments}
\label{supp_sec:experiments}

We conduct experiments over three extensive form games to demonstrate the versatility of the algorithm over n-player general-sum games. For each game we run on both JPSRO(CCE) and JPSO(CE) algorithms under all suitable meta-solvers and baselines.

For JPSRO(CCE), we initialize using uniform policies, update all players at every iteration, and use multi-sets for the pool. For JPSRO(CE), we initialize using uniform policies, update all players at every iteration, only add the highest-gap BR to the pool for each player at each iteration, and use multi-sets for the pool. For random meta-solvers we repeat the experiment five times and show the average, otherwise the experiment is deterministic. The experiments were run for 6 hours, after which any that had not finished were truncated.

In order to measure performance, we track five metrics:
\begin{enumerate}
    \item The gap to equilibrium under a maximum welfare equilibrium (MW(C)CE) distribution. This describes how close the algorithm is to finding a set of joint policies that are in exact equilibrium in the extensive form game.
    \item The gap to equilibrium under the meta-solver's distribution. This is the gap that JPSRO theoretically converges to when using (C)CEs.
    \item The value of the game to the players under the MW(C)CE distribution.
    \item The value of the game to the players under the meta-solver's distribution.
    \item The number of unique policies found so far.
\end{enumerate}

Ultimately, the algorithm should be finding high-value joint policies that are in equilibrium, over a variety of games. The first game is a purely competitive, three-player game called Kuhn Poker (Figure~\ref{fig:joint_psro_kuhn_poker_supp}). The second game is a purely cooperative, common-payoff game called Trade Comm (Figure~\ref{fig:joint_psro_trade_comm_supp}). The final game is a general-sum game called Sheriff (Figure~\ref{fig:joint_psro_sheriff_supp}).

\section{Open Source Code}

An open source implementation of JPSRO is available in OpenSpiel \cite{lanctot2019_openspiel} under \url{https://github.com/deepmind/open_spiel/blob/master/open_spiel/python/examples/jpsro.py}.

\section{Necessity of Population Based Training}

In the absence of a correlating signal, a single joint policy is, in general, insufficient to represent a correlated equilibrium. To see this, let us consider the Traffic Light game (Figure \ref{fig:ce_tl_payoff}). One possible correlated equilibrium consists in recommending (G, W) half of the time, and (W,G) the other half.

Let us now consider this game as an extensive-form, partial-information game, where the row player first chooses their action, and the column player then chooses their own without knowing the action chosen by the row player. In the absence of a correlating signal, it is impossible for the column player to know which action the row player has played, and therefore playing (G, W) or (W, G) becomes impossible, as the column player is unable to change their action as a function of the action taken by the row player.

Therefore, without modifying the game and observation space to add a correlating signal, convergence to a correlated equilibrium necessarily requires a distribution over joint policies. Population Based Training (PBT), a set of methods that slowly grow the space of (joint) policies, therefore appears to be the appropriate framework to converge to (C)CEs without adding correlating signals to the considered game.

\begin{figure*}[p]
    \includegraphics[width=1.\textwidth]{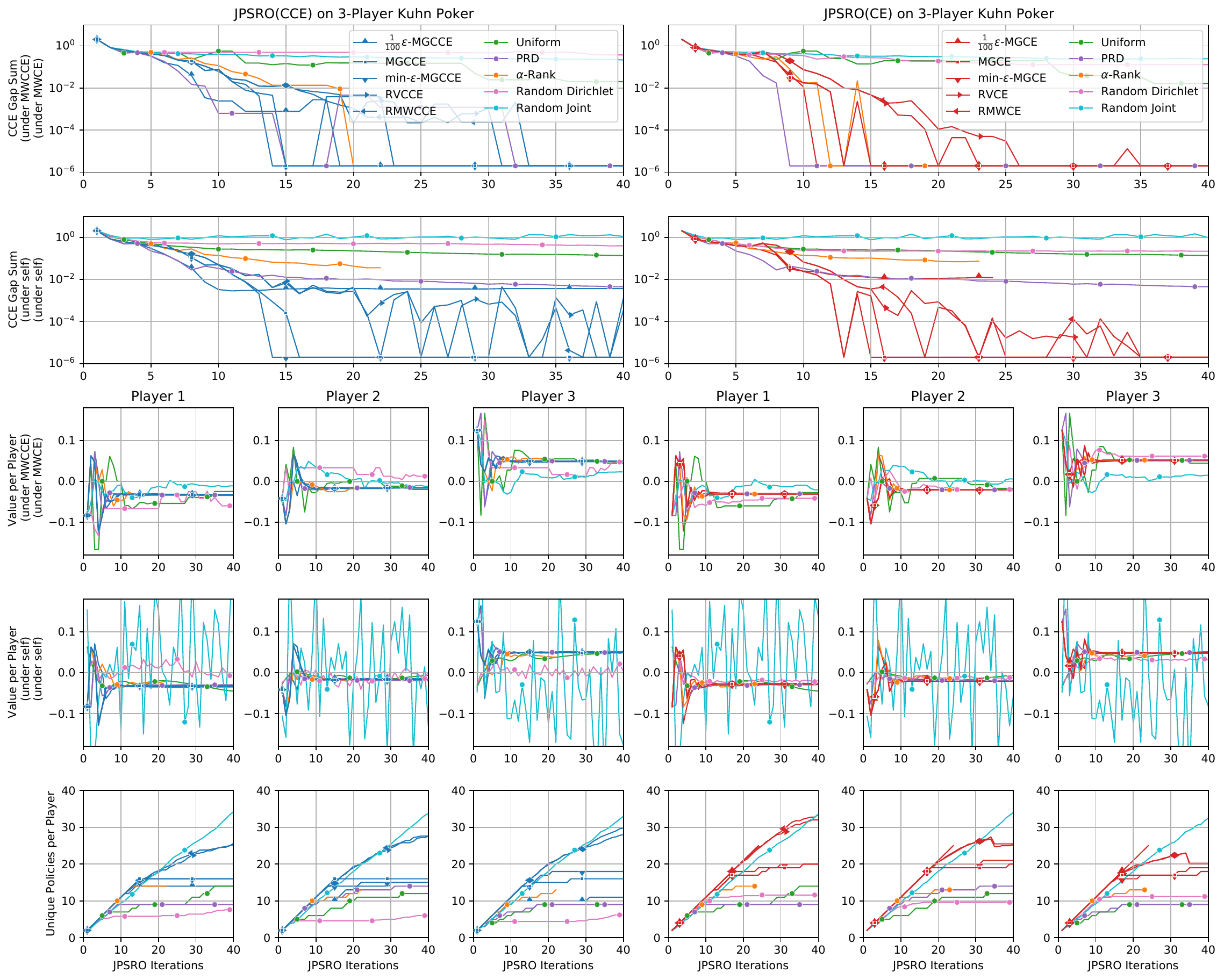}
    \caption{JPSRO(CCE) and JPSRO(CE) on three-player Kuhn Poker. All (C)CE MSs, PRD and $\alpha$-Rank find joint policies capable of supporting equilibrium (although $\alpha$-Rank was slow and was terminated after 6 hours). This is some evidence that classic MSs designed for the two-player, zero-sum setting can generalize well to the three-player, zero-sum.}
    \label{fig:joint_psro_kuhn_poker_supp}
\end{figure*}

\begin{figure*}[p]
    \includegraphics[width=1.\textwidth]{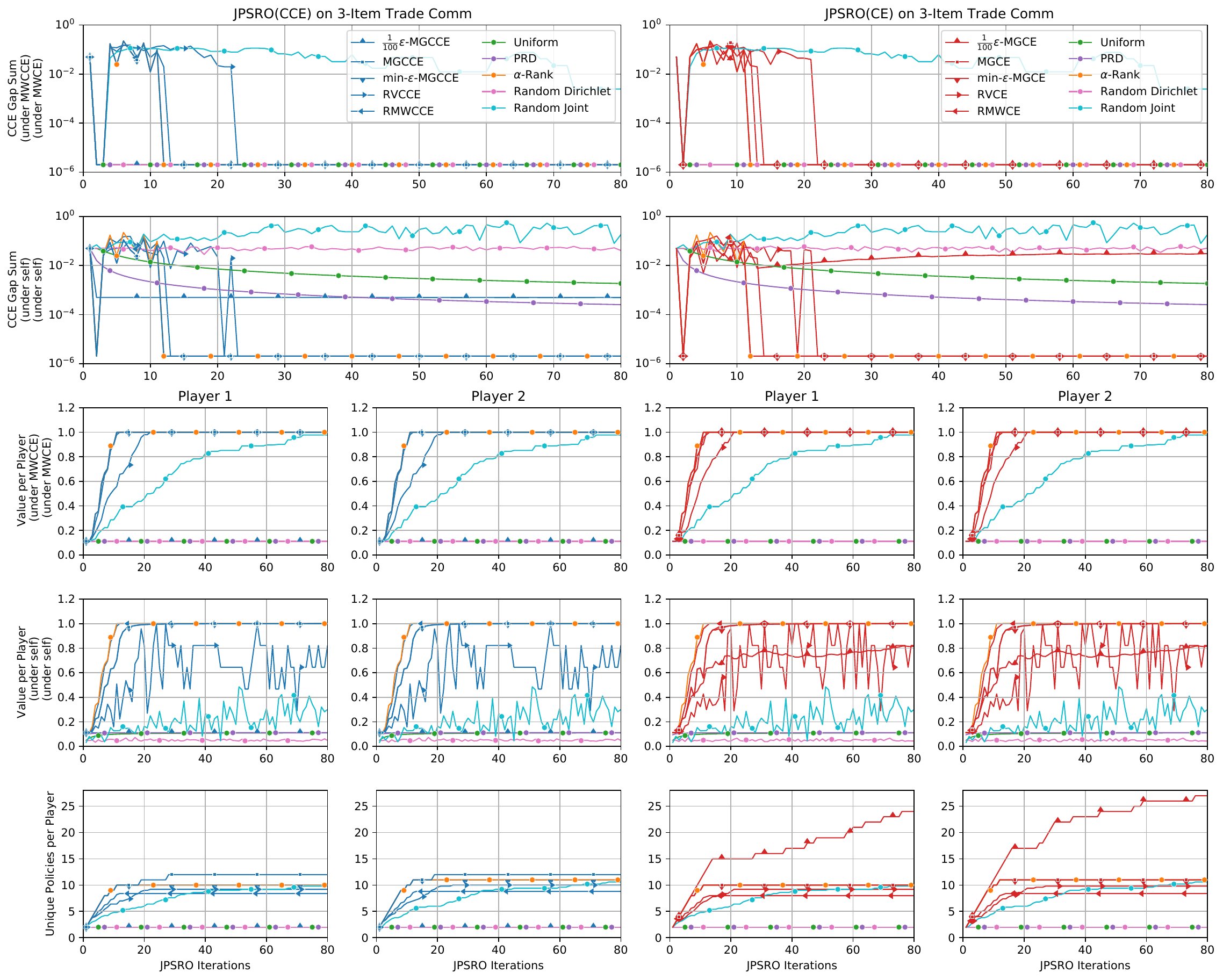}
    \caption{JPSRO(CCE) and JPSRO(CE) on three-item Trade Comm. In JPSRO(CCE), $\frac{1}{100}\min$-MGCCE fails to find the maximum welfare equilibrium, however, all other (C)CE MSs find the maximum welfare equilibrium. Unexpectedly, $\alpha$-Rank performs well on this game, while all other classic MSs fail to make progress on this purely cooperative game. Performing well on this game requires exploration, so the random joint MS is able to make progress, albeit naively and slowly.}
    \label{fig:joint_psro_trade_comm_supp}
\end{figure*}

\begin{figure*}[p]
    \includegraphics[width=1.\textwidth]{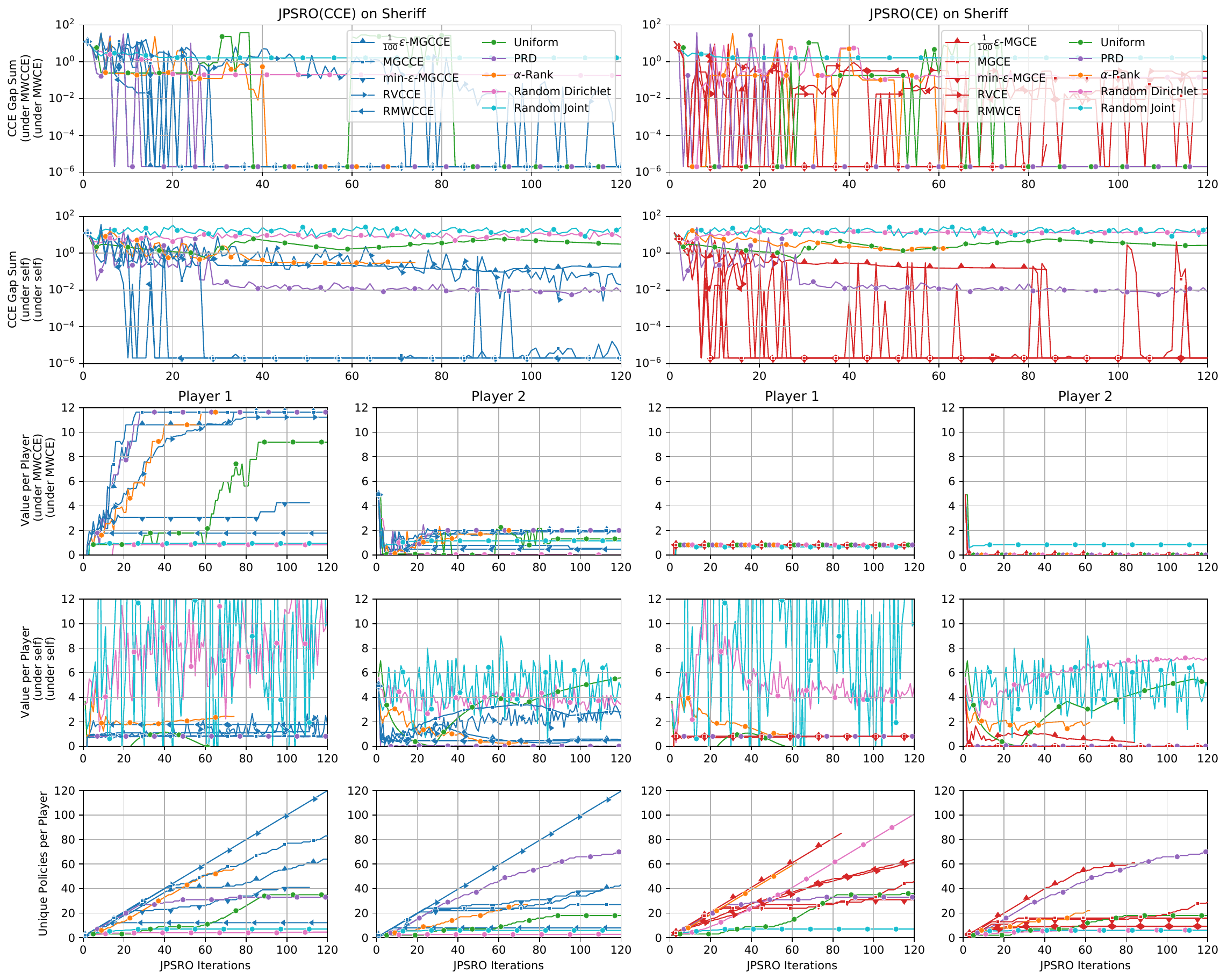}
    \caption{JPSRO(CCE) and JPSRO(CE) on Sheriff. This game is interesting because it is general-sum and different solution concepts have different optimal maximum welfare values. The maximum welfare NFCCE is $13.64$ for the smuggler and $2.0$ for the sheriff which JPSRO(CCE) successfully finds, while the maximum welfare NFCE is $0.82$ for the smuggler and $0.0$ for the sheriff which JPSRO(CE) successfully finds. This demonstrates the appeal of using NFCCE as a target equilibrium. Interestingly, for this game, $\frac{1}{100}\epsilon$-MG(C)CE was able to produce BRs of high enough quality to converge which is evidence that scaled methods that only approximate (C)CEs may be enough in some settings. RMWCCE converged to an equilibrium, but not the welfare maximizing one, providing evidence that greedy MSs are not always suitable. In a similar argument, $\min$-$\epsilon$-MGCCE did not reach the maximum welfare solution within the allocated number of iterations. RV(C)CE is efficient at finding novel policies but ones of limited utility. PRD and $\alpha$-Rank perform well and find the maximum welfare (C)CE equilibria.}
    \label{fig:joint_psro_sheriff_supp}
\end{figure*}

\end{document}